\documentclass{article}
\usepackage[latin9]{inputenc}
\usepackage{amsthm}
\usepackage{amsmath}
\usepackage{amssymb}
\usepackage{graphicx}
\usepackage[unicode=true,pdfusetitle,
 bookmarks=true,bookmarksnumbered=false,bookmarksopen=false,
 breaklinks=false,pdfborder={0 0 1},backref=section,colorlinks=false]
 {hyperref}

\makeatletter
\theoremstyle{plain}
\newtheorem{thm}{\protect\theoremname}
  \theoremstyle{definition}
  \newtheorem{defn}[thm]{\protect\definitionname}
  \theoremstyle{remark}
  \newtheorem{rem}{\protect\remarkname}
  \theoremstyle{plain}
  \newtheorem{prop}[thm]{\protect\propositionname}
  \theoremstyle{plain}
  \newtheorem{lem}[thm]{\protect\lemmaname}

\DeclareMathOperator*{\argmax}{argmax}

\makeatother

  \providecommand{\definitionname}{Definition}
  \providecommand{\lemmaname}{Lemma}
  \providecommand{\propositionname}{Proposition}
  \providecommand{\remarkname}{Remark}
\providecommand{\theoremname}{Theorem}

\title{Identifying the Support of Rectangular Signals in Gaussian Noise}

\author{Jiyao Kou \\ Department of Statistics, Stanford University}

\begin{document}

\maketitle

\begin{abstract}
We consider the problem of identifying the support of the block signal
in a sequence when both the length and the location of the block signal
are unknown. The multivariate version of this problem is also considered, in which we try
to identify the support of the rectangular signal in the hyper-rectangle. 
We allow the length of the block signal to grow polynomially
with the length of the sequence, which greatly generalizes the previous
results in \cite{jeng2010optimal}. A statistical boundary above which the
identification is possible is presented and an asymptotically optimal
and computationally efficient procedure is proposed under Gaussian
white noise in both the univariate and multivariate settings. 
The problem of block signal identification is shown to
have the same statistical difficulty as the corresponding problem
of detection in both the univariate and multivariate cases, in the sense that whenever we can detect the signal, we can identify the support of the signal. Some generalizations are also considered here: (1) We extend our theory to the case of multiple block signals. (2) We also discuss about the robust identification problem when the noise distribution is unspecified and the block signal identification problem under the exponential family setting.

\end{abstract}

\section{Introduction}

Block signal detection and identification in a long one-dimensional
sequence is a challenging and important problem and arises in many
applications, for example, in epidemiology \cite{gangnon2001weighted,neill2009empirical}
and Copy Number Variation \cite{jeng2010optimal,stranger2007relative}.
Block signal detection determines whether there exists any block signal
in the sequence while block signal identification further identifies
the support of the block signal. There has been a large body of work
on signal detection, see, e.g. \cite{glaz2012scan,dumbgen_multiscale_2001}
on the scan statistic; \cite{arias2005near,arias2011detection} for
geometric objects and cluster detection;
\cite{donoho2004higher,delaigle2009higher,donoho2015higher}
for sparse signals detection and identification; \cite{dumbgen2008multiscale,rivera2013optimal}
about density inference and \cite{walther2010optimal,rufibach2010block,chan2009detection,chan2011detection}
for more recent results on block signal detection using the penalized
scan and average likelihood ratio. However, most of the the previous research focus on the univariate case rather than multivariate case and the detection problem rather than the identification problem. \cite{arias2005near} considers the detection of block signal in both the univariate and multivariate  cases, but not the identification. Moreover, the results are actually not optimal unless the size of the block are on the smallest scale. For the block signal identification problem in the univariate case, in \cite{jeng2010optimal},
the authors characterized the identifiable region under the assumption of Gaussian
white noise and  $\log|I_{n}^{*}|=o(\log n)$,
where $n$ is the length of the sequence and $|I_{n}^{*}|$ is the
length of the block signal. They also proposed the Likelihood Ratio Selector
(LRS) procedure and established its optimality under the above assumptions.
However, their result excludes the common and important situation
where $|I_{n}^{*}|=n^{1-\beta}$ for $0<\beta<1$. In fact, it can be shown that LRS is
not optimal in this situation. Moreover, LRS procedure needs to pre-specify
a parameter $L$ which is some number greater than $|I_{n}^{*}|$.
Such $L$ is not always easy to pre-specify and the misspecification
may cause misidentification. 

In this paper, we establish the block signal identification theory
under a more general assumption which includes the case $|I_{n}^{*}|=n^{1-\beta}$
for $0<\beta<1$ in the univariate setting. The multivariate version of this problem is also considered.
A computational efficient procedure based on
the penalized scan statistic is proposed and its optimality is established
under Gaussian white noise assumption in both the univariate and multivariate settings. 
We note that in our procedure, there is
no unknown parameters that need to be pre-specified. Moreover, our
results show that the block signal detection and block signal identification
have the same statistical difficulty in both the univariate and multivariate settings, although the latter
seems to be more challenging than the former.

In addition, we consider in our paper several generalizations of
the block signal identification problem. Firstly, we consider an extension
to the case of multiple block signals. We show that under
certain assumptions, our procedure remains optimal in identifying
all block signals. Moreover, in the discussion section, we briefly
consider the robust identification problem when the noise distribution
is unspecified and discuss about the block signal identification
under the exponential family setting.

The rest of the paper is organized as follows. In Section \ref{sec:Block-identification-simple}, our identification procedure is introduced in the univariate case and its optimality is established under Gaussian white noise. In Section \ref{sec:Multi-dimensional-block-signal},
we extend our theorem to the multi-dimensional case and we show that our
procedure remains optimal in identifying rectangular signals in the hyper-rectangle. In Section \ref{sec:Signal-identification-More-Than-One},
we consider the situation when there are multiple block signals. In
Section \ref{sec:Simulation-Study}, a simulation study is
carried out to illustrate our previous results. In Section \ref{sec:Discussion},
we give a brief discussion about the identification under an unknown
noise distribution and under the exponential family setting. We also discuss
about some future research topics.

In the end of this section, we make some notations. For two series $a_n$ and $b_n$, we define $a_n \ll b_n$ if $a_n = o(b_n)$, or equivalently, $\frac{a_n}{b_n} \rightarrow 0$. We may use this notation and the small-o notation interchangeably. For a set of random variables $X_{n}$ and a corresponding set of constant $a_{n}$, we define $X_{n}=o_{p}(a_{n})$ if $X_{n}/a_{n}$ converges to 0 in probability. Similarly, we define $X_{n}=O_{p}(a_{n})$ if for any $\epsilon>0$, there exists a finite $M$ such that $P(|X_{n}/a_{n}|>M)<\epsilon$ for all $n$.

\section{Block signal identification under Gaussian white noise\label{sec:Block-identification-simple}}
Let's first consider block signal identification in the univariate setting. Suppose we observe that 
\begin{equation}
Y_{i}=\mu\mathbf{1}_{I_{n}^{*}}(i)+Z_{i},\qquad i=1,\ldots,n\label{eq:model}
\end{equation}
where $Z_{i}$ are i.i.d standard normal random variables and unknown
interval $\, I_{n}^{*}=(j_{n},k_{n}],\enskip0\le j_{n}<k_{n}<n$ and
\[
\boldsymbol{1}_{I_{n}^{*}}(i)=\begin{cases}
1 & i\in I_{n}^{*}\\
0 & i\notin I_{n}^{*}
\end{cases}.
\]
$\mu$ is an unknown number and for simplicity we assume that $\mu$ is non-negative. If $\mu$ is non-positive, we can replace $Y_{i}$ by $-Y_{i}$. Our goal is to estimate the support $I_{n}^{*}$ in model \eqref{eq:model} where $I^{*}_n=\emptyset$ means $\mu=0$.

For block signal \emph{detection}, which is a testing problem, we want to maximize the power of the
test while controlling the type I error. Similarly, in the corresponding block signal \emph{identification}
problem we want to approximately find the start and the end point
of the block signal (when exists) with high probability while control the type I error. To give the
definition of consistency, we introduce the following notation.
Let $H_{0}$ denote the null case that there exists no signal
in the sequence and $H_{1}$ denote the case where there exists a block signal $I^{*}_n$.
Define the (Hamming) distance between two intervals $I_{1}$ and $I_{2}$
as $D(I_{1},I_{2})=1-\frac{|I_{1}\cap I_{2}|}{\sqrt{|I_{1}||I_{2}|}}$,
where we take the convention $\frac{0}{0}=0$. The definition of consistency
for the block signal identification problem is given below.
\begin{defn}
We call a procedure $\mathcal{P}$ to be consistent if its estimated
interval $\hat{I}_{n}$ satisfies 
\begin{equation}
P_{H_{0}}(\hat{I}_{n}\ne\emptyset)\le\alpha\label{eq:type I error}
\end{equation}
and
\begin{equation}
P_{H_{1}}(D(\hat{I}_{n},I_{n}^{*})>\delta_{n})\rightarrow0\label{eq: type II error}
\end{equation}
for some $\delta_{n}=o(1)$, where $\emptyset$ denotes the empty
set and $\alpha$ denotes the significance level.
\end{defn}

In this section, we focus on those $I_{n}^{*}$ satisfying the following
property: there exists a $\kappa>0$, such that $|I_{n}^{*}|\ll n^{1-\kappa}$. This mild assumption includes all intervals with length $n^{1-\beta}$
for $0<\beta\le1$, but not those with length $n/\log n$. The set
of intervals considered here greatly extends those in \cite{jeng2010optimal},
in which the author requires $\log(|I_{n}^{*}|)\ll\log n$.

Before giving the identification procedure, we first introduce the concept of the
approximation set which is introduced in  \cite{rivera2013optimal, walther2010optimal, chan2011detection}. The idea of the approximation set is that we only need to consider
intervals with endpoints on a grid as long as we can approximate each
interval relatively well. In this section, we define our approximation set as below:
\[
\mathcal{I}_{app}=\bigcup_{\ell=1}^{\ell_{max}}\mathcal{I}_{app}(\ell)\cup\mathcal{I}_{small},\quad\text{where }\ell_{max}=\lfloor\log_{2}\frac{n}{\log n}\rfloor
\]
in which,
\[
\mathcal{I}_{app}(\ell)=\left\{ (j,k]:j,k\in\{id_{\ell},i=0,1,\ldots\}\text{ and }m_{\ell}<k-j\le2m_{\ell}\right\} ,
\]
\[
\mathcal{I}_{small}=\{(j,k]:k-j\le m_{\ell_{max}}\}.
\]
where $m_{\ell}=n2^{-\ell}$, $d_{\ell}=\lceil\frac{m_{\ell}}{6\sqrt{\ell}}\rceil$.
A simple counting argument shows that $|\mathcal{I}_{app}(\ell)| \le (\frac{n}{d_\ell}+1) (\frac{m_\ell}{d_\ell}+1) \le 144 \ell 2^{\ell}$ for $\ell = 1, \ldots, \ell_{max}$ and $|\mathcal{I}_{small}| \le 2n\log n$. Thus $|\mathcal{I}_{app}| \le \sum_{\ell=1}^{\ell_{max}} 144\ell 2^{\ell}+ 2n\log n = O(n\log n)$.

\begin{rem}

1) In fact we can let $d_{l}=\lceil\frac{m_{\ell}}{c\ell^{\zeta}}\rceil$
for some $c>0$ and $\zeta\ge0.5$. $c$ and $\zeta$ control the
precision of the approximation set and the choice is a trade off between
computational efficiency and approximation error: the larger the $c$
and $\zeta$, the better the approximation while the heavier the computation.

2) Some different approximation sets are also introduced in \cite{neill2004detecting} and \cite{arias2005near}. It is not clear whether those approximation sets can lead to the same optimal result. 

\end{rem}

Define $\boldsymbol{Y}(I)=\frac{\sum_{i \in I} Y_i}{\sqrt{|I|}}$. Our identification procedure, denoted by $\mathcal{P}_{n}$, works
as follows: If $\max_{I\in\mathcal{I}_{app}}(\boldsymbol{Y}(I) - \sqrt{2\log\frac{en}{|I|}})<\gamma_{n}(\alpha)$,
where $\gamma_{n}(\alpha)$ is the $(1-\alpha)$ quantile of
the null distribution of $\max_{I\in\mathcal{I}_{app}}(\mathbb{\mathbf{Y}}(I)-\sqrt{2\log\frac{en}{|I|}})$,
we claim there exists no signal, i.e. $\hat{I}_{n}=\emptyset$. Otherwise,
our estimated interval is
\begin{equation}
\hat{I}_{n}=\argmax_{I\in\mathcal{I}_{app}}\left(\boldsymbol{Y}(I)-\sqrt{2\log\frac{en}{|I|}}\right).\label{eq: Null-lemma}
\end{equation}
By Boole's inequality, we can show that under the null distribution, 
$P_n := \max_{I\in\mathcal{I}_{app}}\left(\boldsymbol{Y}(I)-\sqrt{2\log\frac{en}{|I|}}\right) = O_p(1)$, so $\limsup_{n \rightarrow \infty} \gamma_n(\alpha)<\infty$, see \cite{rivera2013optimal}.

It can be shown that $\mathcal{P}_{n}$ is optimal for block signal
identification. In fact, the procedure $\mathcal{P}_{n}$ is consistent
in identifying the support $I_{n}^{*}$ whenever the signal is in the
detectable region. We summarize this fact in the following theorem.
\begin{thm}
\label{thm: Main theorem}Assume Model \eqref{eq:model} and there
exists a $\kappa>0$ such that $|I_{n}^{*}|\ll n^{1-\kappa}$. If
$\mu\ge(\sqrt{2\log\frac{en}{|I_{n}^{*}|}}+b_{n})/\sqrt{|I_{n}^{*}|}$
with $b_{n}\rightarrow+\infty$, then our identification
procedure $\mathcal{P}_{n}$ is consistent with any $1\gg\delta_{n}\gg \sqrt{\log\log n} / \sqrt{\log{n}}$.
In addition, this procedure can be computed in $O(n\log n)$ time. 
\end{thm}

From Section 2 of \cite{chan2011detection} , $\mu\ge(\sqrt{2\log\frac{en}{|I_{n}^{*}|}}+b_{n})/\sqrt{|I_{n}^{*}|}$
with $b_{n}\rightarrow+\infty$ is necessary for any test to be consistent
in detecting the signal $I_{n}^{*}$. Since the identification problem
is more challenging than the corresponding detection problem, we can
conclude that $\mu\ge(\sqrt{2\log\frac{en}{|I_{n}^{*}|}}+b_{n})/\sqrt{|I_{n}^{*}|}$
with $b_{n}\rightarrow+\infty$ is necessary for any procedure to
be consistent in identifying the signal $I_{n}^{*}$. Thus our procedure
$\mathcal{P}_{n}$ is in fact optimal in block signal identification
under our current setting. 

We can view $\delta_n$ as the precision for the signal recovery. From Theorem
\ref{thm: Main theorem}, we know that for any $\delta_n = \phi_n \sqrt{\log\log n} / \sqrt{\log{n}}$ with $\phi_n \rightarrow \infty$ at any rate, we have 
\[
P_{H_{1}}(D(\hat{I}_{n},I_{n}^{*})<\delta_{n}) \rightarrow 1.
\]
Thus $\delta_n = \phi_n \sqrt{\log\log n} / \sqrt{\log{n}}$ is the upper bound for the (hamming) distance between our recovered signal and the true signal. As $\delta_n \rightarrow 0$, we can be sure that our recovered signal is sufficiently close to the true signal asymptotically.

We also note that the block signal identification
problem has the same statistical difficulty as the corresponding detection
problem, although the former one seems to be more challenging than
latter one. In the computational aspect, although there are $O(n^{2})$
number of possible intervals in the sequence, our procedure runs in
$O(n\log n)$ time, almost linear in the number of observations.
\begin{rem}

1) If we let $d_{l}=\lceil\frac{m_{\ell}}{c\ell^{\zeta}}\rceil$ for
some $c>0$ and $\zeta\ge0.5$ in the definition of approximation set, then Theorem \ref{thm: Main theorem} still holds and the computational complexity is $O(nc^2 \log^{2\zeta}n)$. 

2) If instead we define $\gamma_{n}(\alpha)$ as the $(1-\alpha)$
quantile of the null distribution of $\max_{0\le j<k\le n}\left(\boldsymbol{Y}((j,k])-\sqrt{2\log\frac{en}{k-j}}\right)$,
Theorem \ref{thm: Main theorem} would still hold with the same bound
on $\delta_{n}$. However, one would expect a faster simulation for
critical values using the previous definition of $\gamma_{n}(\alpha)$.
\end{rem}

We make the following comparison between our procedure and the LRS
procedure in \cite{jeng2010optimal}. (1) The LRS procedure requires the
length of the signal $|I_{n}^{*}|$ to satisfy $\log|I_{n}^{*}|=o(\log n)$.
In contrast, we allow $|I_{n}^{*}|=n^{1-\beta}$ for $0<\beta\le1$.
(2) The identification boundary for LRS procedure is $\frac{\sqrt{2\log n}}{\sqrt{I_{n}^{*}}}$,
which is optimal only for signals with the smallest spatial extent.
In contrast, the identification boundary for $\mathcal{P}_{n}$ is $\sqrt{2\log\frac{en}{|I_{n}^{*}|}}/\sqrt{|I_{n}^{*}|}$, which is optimal for a broad range of signals. (3) In the LRS procedure,
one needs to pre-specify the parameter $L$ which is some number greater
than the length of the signal. However, since the length of the signal
is unknown, it is not always easy to determine $L$ and the misspecification
of $L$ may cause misidentification. Although it is argued in \cite{jeng2010optimal} that
$L$ could be sometimes easily selected, we would prefer a procedure without
any unknown parameters. Our procedure $\mathcal{P}_{n}$ has no need
to specify the parameter $L$ or any other unknown parameters. (4)
The LRS procedure has a computational complexity of $O(nL)$. Depending
on the choice of $L$, the complexity could be large. In contrast, our
procedure $\mathcal{P}_{n}$ has a computational complexity of $O(n\log n)$,
regardless of the length of the signal. We also note that besides LRS
procedure in \cite{jeng2010optimal}, similar multiscale methods also appear in \cite{arias2005near}, \cite{arias2011detection} and \cite{neill2004detecting}. However, these methods do not lead to the same optimal results in our setting.

\section{\label{sec:Multi-dimensional-block-signal}Multi-dimensional rectangular
signal identification}

In this case, we consider the problem of identifying the rectangular
signal in the multi-dimensional hyper-rectangle. 

Consider the $D$-dimensional model
\begin{equation}
Y_{i}=\mu\mathbf{1}_{I_{n}^{*}}(i)+Z_{i},\qquad i=(i^{1},\ldots,i^{D})\in[1,\ldots,n]^{D}\label{eq:model-multi-dimension}
\end{equation}
where $Z_{i}$ are i.i.d standard normal random variable and unknown
rectangle 
\[
\enskip I_{n}^{*}=\prod_{d=1}^{D}(j_{n,d},k_{n,d}],\quad0\le j_{n,d}<k_{n,d}<n,\quad\text{for each }d=1,\ldots D
\]
with sides parallel to the axes and with arbitrary sizes and aspect
ratios. We denote the area of the hyper-rectangle by $|I_{n}^{*}|=\prod_{d=1}^{D}(k_{n,d}-j_{n,d})$.

All of the results in this section can be easily extended to higher
dimensions, but we will focus on the two-dimensional case $D=2$ to
simplify our notation. We use the superscript $\{2\}$ to denote we
are now considering the two-dimensional case.

We first introduce the approximation set for two-dimensional rectangle, which is a variation of the construction in \cite{walther2010optimal}.
Using $(j_{1},j_{2},k_{1},k_{2})$ denotes the rectangle with vertices
$(j_{1},k_{1})$, $(j_{2},k_{1})$, $(j_{2},k_{2})$ and $(j_{1},k_{2})$.
For fixed $\ell$ and $0\le i\le\ell$, define 
\begin{eqnarray*}
\mathcal{I}_{app}^{\{2\}}(\ell,i) & = & \{(j_{1},j_{2},k_{1},k_{2}):\enskip j_{1},j_{2}\enskip\in\enskip\{r^{(1)}d_{\ell,i}^{(1)},r^{(1)}=0,1,\ldots\}\\
 &  & \text{ and }k_{1},k_{2}\enskip\enskip\in\enskip\{r^{(2)}d_{\ell,i}^{(2)},r^{(2)}=0,1,\ldots\},\\
 &  & 0\le j_{1}\le\lfloor\ell^{\frac{1}{2}}2^{\ell-i}\rfloor,1\le j_{2}-j_{1}\le\lfloor\ell^{\frac{1}{2}}\rfloor,\\
 &  & 0\le k_{1}\le\lfloor\ell^{\frac{1}{2}}2^{i}\rfloor,1\le k_{2}-k_{1}\le\lfloor2\ell^{\frac{1}{2}}\rfloor\}
\end{eqnarray*}
where $d_{\ell,i}^{(1)}=\lceil\ell^{-\frac{1}{2}}2^{-\ell}2^{i}n\rceil$, $d_{\ell,i}^{(2)}=\lceil\ell^{-\frac{1}{2}}2^{-i}n\rceil$. Then
let $\mathcal{I}_{app}^{\{2\}}(\ell)=\cup_{i=0}^{\ell}\mathcal{I}_{app}^{\{2\}}(\ell,i).$
\\
Define $\ell_{max}=\lfloor\log_{2}\frac{n^{2}}{\log n}\rfloor$,
$m_{l}=n2^{-\ell}$ and 
\begin{eqnarray*}
\mathcal{I}_{small}^{\{2\}} & = & \{I,|I|\le m_{\ell_{max}}\}.
\end{eqnarray*}
Then our approximation set $\mathcal{I}_{app}^{\{2\}}$ is defined as: 
\[
\mathcal{I}_{app}^{\{2\}}=\bigcup_{\ell=1}^{\ell_{max}}\mathcal{I}_{app}^{\{2\}}(\ell)\cup\mathcal{I}_{small}^{\{2\}}.
\]
Define $\boldsymbol{Y}(I)=\frac{\sum_{i\in I}Y_{i}}{\sqrt{|I|}}$,
now we are ready to introduce the property of the penalized scan statistic
in the two-dimensional case in the following proposition. 
\begin{prop}
\label{prop: High-dim-null-distribution}
Define
\[
P_{n}^{\{2\}}=\max_{I\in\mathcal{I}_{app}^{\{2\}}}\left(\boldsymbol{Y}(I)-\sqrt{2\log\frac{en^{2}}{|I|}}\right),
\]
Under the null hypothesis, $P_{n}^{\{2\}} = O_p(1)$, i.e. $P_{n}^{\{2\}}$ are uniformly bounded in probability.
\end{prop}
We define $\gamma_{n}^{\{2\}}(\alpha)<\infty$ as the $(1-\alpha)$
quantile of the null distribution of $P_{n}^{\{2\}}$, which is well
defined by Proposition \ref{prop: High-dim-null-distribution}.

The identification procedure $\mathcal{P}_{n}^{\{2\}}$ in the two-dimensional
case works as follows: If $\max_{I\in\mathcal{I}_{app}^{\{2\}}}(\mathbb{\mathbf{Y}}(I)-\sqrt{2\log\frac{en^{2}}{|I|}})<\gamma_{n}^{\{2\}}(\alpha)$,
we claim there is no signal, i.e. $\hat{I}_{n}=\emptyset$. Otherwise,
our estimated rectangle
\[
\hat{I}_n=\argmax_{I\in\mathcal{I}_{app}^{\{2\}}}\left(\boldsymbol{Y}(I)-\sqrt{2\log\frac{en^{2}}{|I|}}\right).
\]
As in the one-dimensional case, one can establish the optimality for
the above procedure in the two-dimensional case, which is given in
the following theorem.
\begin{thm}
\label{thm: Main theorem-high-dim}Assume Model \eqref{eq:model-multi-dimension}
and there exists a $\kappa>0$ such that $|I_{n}^{*}|\ll n^{2-\kappa}$.
If $\mu\ge(\sqrt{2\log\frac{en^{2}}{|I_{n}^{*}|}}+b_{n})/\sqrt{|I_{n}^{*}|}$
with $b_{n}\rightarrow+\infty$, then our
procedure $\mathcal{P}_{n}^{\{2\}}$ is consistent with $1\gg\delta_{n}\gg\sqrt{\log\log n}/\sqrt{\log n}$.
In addition, this procedure can be computed in $O(n^{2}\log^{2}n)$
time. 
\end{thm}
For a very similar argument as in Section 2 of \cite{chan2011detection}
, we can get $\mu\ge(\sqrt{2\log\frac{en^{2}}{|I_{n}^{*}|}}+b_{n})/\sqrt{|I_{n}^{*}|}$
is necessary for any test to be consistent in detecting the rectangular
signal in the two-dimensional case and as a result $\mu\ge(\sqrt{2\log\frac{en^{2}}{|I_{n}^{*}|}}+b_{n})/\sqrt{|I_{n}^{*}|}$
is also necessary for any procedure to be consistent in identifying
the signal in the two-dimensional case. Thus our procedure is optimal
under our current setting. Again, we conclude that in the two-dimensional
case, the identification problem and the corresponding detection problem
has the same statistical difficulty. Although there are $O(n^{4})$
number of possible rectangles, our algorithm runs in $O(n^{2}\log^{2}n)$,
almost linear in the number of observations $n^{2}$. 
In general, in $d$ dimensional case, our algorithm runs in  $O(n^{d}\log^{d}n)$,
almost linear in the number of observations $n^{d}$.

\section{\label{sec:Signal-identification-More-Than-One}Signal identification
for multiple signals}

In the previous two sections, we focus on the situation where there
exists only one signal. In this section, we consider the situation
where there are multiple signals. We will only discuss the univariate case in this section for notation simplicity, but all our theory can be extended to multivariate case 
by using the corresponding approximation set. To further simplify the notation and
avoid confusion, we suppress the subscript $n$ in $I_{n}^{*}$ in
this section. We denote the set of true block signals $I^{*}=\{I_{1}^{*},\ldots,I_{K}^{*}\}$,
where $K$ is the number of block signals. Our model is:

\begin{equation}
Y_{i}=\mu_{I^{*}}(i)+Z_{i},\qquad i=1,\ldots,n\label{eq:model-multiple signal}
\end{equation}
where $Z_{i}$ are i.i.d standard normal random variable, the unknown
set of intervals $I^{*}=\{I_{1}^{*},\ldots,I_{K}^{*}\}=\{(j_{1},k_{1}],\ldots,(j_{K},k_{K}],\,0\le j_{1}<k_{1}<j_{2}<\ldots<k_{K}\le n\}$
and 
\[
\mu_{I^{*}}(i)=\begin{cases}
\mu_{I_{j}^{*}} & i\in I_{j}^{*}\\
0 & \text{otherwise}
\end{cases}.
\]
In another word, $\mu_{I^{*}}(i)$ has constant value $\mu_{I_{j}^{*}}$ on each interval
$I_{j}^{*}$, $j=1,\ldots K$ and 0 otherwise.

We give the following definition of consistency for the identification
procedure for multiple block signals. 
\begin{defn}
\label{Def: Multiple signal Consistency}We call a procedure $\mathcal{P}$
is consistent in identifying multiple block signals if its estimated set of signals (intervals) $\hat{I}=\{\hat{I}_{1},\ldots,\hat{I}_{\hat{K}}\}$
satisfies 
\begin{equation}
P_{H_{0}}(\hat{I}\ne\emptyset)\le\alpha\label{eq: Type I of multiple signals}
\end{equation}
\begin{equation}
E_{H_{1}}(\hat{K})\le K+C(\alpha) + o(1) \label{eq: Expect of multiple signals}
\end{equation}
and
\begin{equation}
P_{H_{1}}(\max_{j=1}^{K}\min_{i = 1}^{\hat{K}} D(\hat{I}_{i},I_{j}^{*})>\delta_{n})\rightarrow0\label{eq: Consistency of multiple signals}
\end{equation}
for some $\delta_{n}=o(1)$, where $\emptyset$ denotes the empty
set, $\alpha$ denotes the significant level, $C(\alpha)>0$ is a
function that depends on $\alpha$ and $\lim_{\alpha\rightarrow0}C(\alpha)=0$.
\end{defn}
Equation \eqref{eq: Consistency of multiple signals} in fact implies that
$P_{H_{1}}(\hat{K}\ge K)\rightarrow1$ as $n\rightarrow\infty$. Together
with Equation \eqref{eq: Expect of multiple signals}, we can conclude that
if a procedure is consistent asymptotically, it can identify all true
intervals and may include some false intervals. In expectation, the
number of the false intervals our procedure identifies goes to 0 as
$\alpha\rightarrow0$ and $n\rightarrow\infty$. 

Note this slightly complicated definition reflects a fundamental difficulty
in this problem: even after we correctly identifying all signals,
we get back to the null cases and we cannot avoid a probability less
than $\alpha$ that we include a false interval. One can take $\alpha$
to be small and Equation \eqref{eq: Type I of multiple signals}, \eqref{eq: Expect of multiple signals}
and \eqref{eq: Consistency of multiple signals} still hold so this
effect is minimal.

Theorem \ref{thm: Main theorem} can be generalized to the current
situation as long as $K=O(\log^{p}n)$ for some $p>0$. Here we also
assume the minimum distance between two signals $d_{min}\gg\max_{j=1}^{K}|I_{j}^{*}|\log n$. 

Define $\mathcal{I}_{app}$ and $\gamma_{n}(\alpha)$ as in Section
\ref{sec:Block-identification-simple}. Our identification procedure,
denote by $\mathcal{P}_{n,multi}$ for multiple signals works as follows:

Initialize our result set $\hat{I}$ as empty set $\emptyset$. Denote
$\mathcal{I}_{app}^{1}=\mathcal{I}_{app}$.

Let $i=1$, repeat the following step until $\max_{I\in\mathcal{I}_{app}^{i}}(\mathbb{\mathbf{Y}}(I)-\sqrt{2\log\frac{en}{|I|}})<\gamma_{n}(\alpha)$:

\{

$\qquad$Denote $\hat{I}_{i}=\text{\ensuremath{\arg}max}_{I\in\mathcal{I}_{app}^{i}}(\mathbb{\mathbf{Y}}(I)-\sqrt{2\log\frac{en}{|I|}})$,
then set $\hat{I}=\hat{I}\cup\hat{I}_{i}$,

$\qquad$Let $\mathcal{I}_{app}^{i+1}=\mathcal{I}_{app}^{i}\backslash\{I\in\mathcal{I}_{app}^{i}:I\cap\hat{I}_{i}\ne\emptyset\}$.

$\qquad$Increase $i$ by 1.

\}

We have the following theorem regarding the consistency of the procedure
$\mathcal{P}_{n,multi}$.
\begin{thm}
\label{thm: More than one signal}Assume Model \eqref{eq:model-multiple signal} and there exists a $\kappa>0$ such that $|I_j^*| \ll n^{1-k}$ for each $j=1,\ldots,K$ and the minimum distance between two signals $d_{min}\gg\max_{j=1}^{K}|I_{j}^{*}|\log n$ and $K=O(\log^{p}n)$ for some $p>0$, if for each $j=1,\ldots,K$, $\mu_{I_{j}^*}\ge(\sqrt{2\log\frac{en}{|I_{j}^{*}|}}+b_{n,j})/\sqrt{|I_{j}^{*}|}$
with $b_{n,j}/\sqrt{\log\log n}\rightarrow+\infty$,
then our procedure $\mathcal{P}_{n,multi}$ is consistent with $1\gg\delta_{n}\gg  \sqrt{\log\log n}/\sqrt{\log n}$. 
\end{thm}
The main difference here compare to Theorem \ref{thm: Main theorem}
is we require $b_{n,j}/\sqrt{\log\log n}\rightarrow\infty$ rather
than $b_{n,j}\rightarrow\infty$. This stronger requirement will ensure
that the probability of making mistakes for every iteration is small
enough so that the identification procedure is still consistent after
all iterations. The reason we assume that the minimum spacing between the intervals 
$d_{min}\gg\max_{j=1}^{K}|I_{j}^{*}|\log n$ is that if two intervals are very close, it is difficult to identify the exact location of the signal under the influence of another signal. For example, with non-diminishing probability that two intervals both with length $|I^*|$ and distance $|I^*|$ will be identified as one signal with length about $3|I^*|$. This is different from the detection problem where no such requirement is needed.

As a special case, if all signals are of the same length $|I_{1}^{*}|$,
an explicit lower bound can be given.
\begin{thm}
\label{thm:(Lower-bound)-Multi} Assume Model \eqref{eq:model-multiple signal},
and $K=O(\log^{p}n)$, if all signals are of the same length $|I_{1}^{*}|$,
then no procedure is consistent if $\mu_{I_{j}^{*}}<\sqrt{2\log\frac{en}{|I_{1}^{*}|}}/\sqrt{|I_{1}^{*}|}$.
Thus our procedure $\mathcal{P}_{n,multi}$ is in fact optimal under
the above assumptions.
\end{thm}

When $\mu_{I_{j}^{*}}<\sqrt{2\log\frac{en}{|I_{1}^{*}|}}/\sqrt{|I_{1}^{*}|}$, then an algorithm would either fail to discover all signals and/or will include too may false discoveries, see Lemma 1 in \cite{jeng2010optimal} and more details in Theorem 5 in \cite{genovese2009revisiting}.

\begin{rem}
If the true number of signals $K$ is known, we can stop the procedure
$\mathcal{P}_{n,multi}$ after $K$ iterations. It can be shown that
this modified procedure is still consistent and its estimated set does
not contain any false intervals.
\end{rem}

\section{Simulation Study\label{sec:Simulation-Study}}

We have shown that by adopting the correct penalty term, the procedure
based on the penalized scan can be much more powerful, which is the
major difference between our procedure and the LRS procedure in \cite{jeng2010optimal}.
It would be difficult to compare these two procedures directly since
LRS needs to specify additional parameter $L$ and its optimality is established under some additional assumptions. It is also unclear how LRS works in high-dimension. 
So in this section, we will illustrate the
previous results with a simulation study for Model \eqref{eq:model} by comparing the performance of
the identification procedure $\mathcal{P}_{n}$ defined in Section
\ref{sec:Block-identification-simple} and the identification procedure
without the penalty term, denote by $\mathcal{P}_{n}^{unpen}$. To
be comparable, $\mathcal{P}_{n}^{unpen}$ is defined exactly the same
on the same approximation set except we set the penalty term to 
0 instead of $\sqrt{2\log\frac{en}{|I|}}$. It can be shown that the
identification boundary for $\mathcal{P}_{n}^{unpen}$ is $\sqrt{2\log n}/\sqrt{|I_{n}^{*}|}$, which is the same as LRS.
If $|I_{n}^{*}|=n^{1-p}$, $p\in(0,1)$, this detection boundary is
$p^{-1/2}$ times larger than the optimal boundary. 

Denote the (Hamming) similarity between interval $I_1$ and $I_2$ by $1-D(I_1,I_2)=\frac{|I_1 \cap I_2|}{\sqrt{|I_1||I_2|}}$
and $\gamma_{n}(\alpha)$ as in Section \ref{sec:Block-identification-simple}.
For a particular simulation, if $\max_{I\in\mathcal{I}_{app}}(\boldsymbol{Y}(I)-\sqrt{2\log\frac{en}{|I|}})>\gamma_{n}(\alpha)$,
the similarity is calculated as $\frac{|\hat{I}_n \cap I^{*}_n|}{\sqrt{|\hat{I}_n||I^{*}_n|}}$,
where $\hat{I}_n$ is defined in \eqref{eq: Null-lemma}; otherwise
the similarity is set to be 0. We do the same for $\mathcal{P}_{n}^{unpen}$,
except that $\gamma_{n}(\alpha)$ is replaced by $\tau_{n}(\alpha)$,
which is defined as the $(1-\alpha)$ quantile of $\max_{I\in\mathcal{I}_{app}}\boldsymbol{Y}(I)$.
In all of our simulations for univariate setting, we choose $n=10000$.

For the first simulation, we give the similarity for different choices
of $\mu$ range from 1.5 to 5 with a step of 0.5 and signal length
$|I_{n}^{*}|=100$ and $|I_{n}^{*}|=1000$, respectively. The result
is given in Figure \ref{fig: Similarity vs mu}. From Figure \ref{fig: Similarity vs mu},
we can see that when $|I_{n}^{*}|=1000$, which is relatively large,
$\mathcal{P}_{n}$ performs much better than $\mathcal{P}_{n}^{unpen}$,
while when $|I_{n}^{*}|=100$, which is relatively small, the performance
of $\mathcal{P}_{n}$ is only slightly better. 

\begin{figure}
\caption{\label{fig: Similarity vs mu}Simulated similarities for different
$\mu$. The left sub-graph gives the result for $|I_{n}^{*}|=1000$
and the right sub-graph gives the result for $|I_{n}^{*}|=100$. In
both graphs, the black solid line denotes the penalized procedure
$\mathcal{P}_{n}$ and the right dashed line denotes the unpenalized
procedure $\mathcal{P}_{n}^{unpen}$. The $x$-axis is $\mu \sqrt{|I^*|}$ and $y$-axis is similarity.}

\begin{centering}
\includegraphics[scale=0.30]{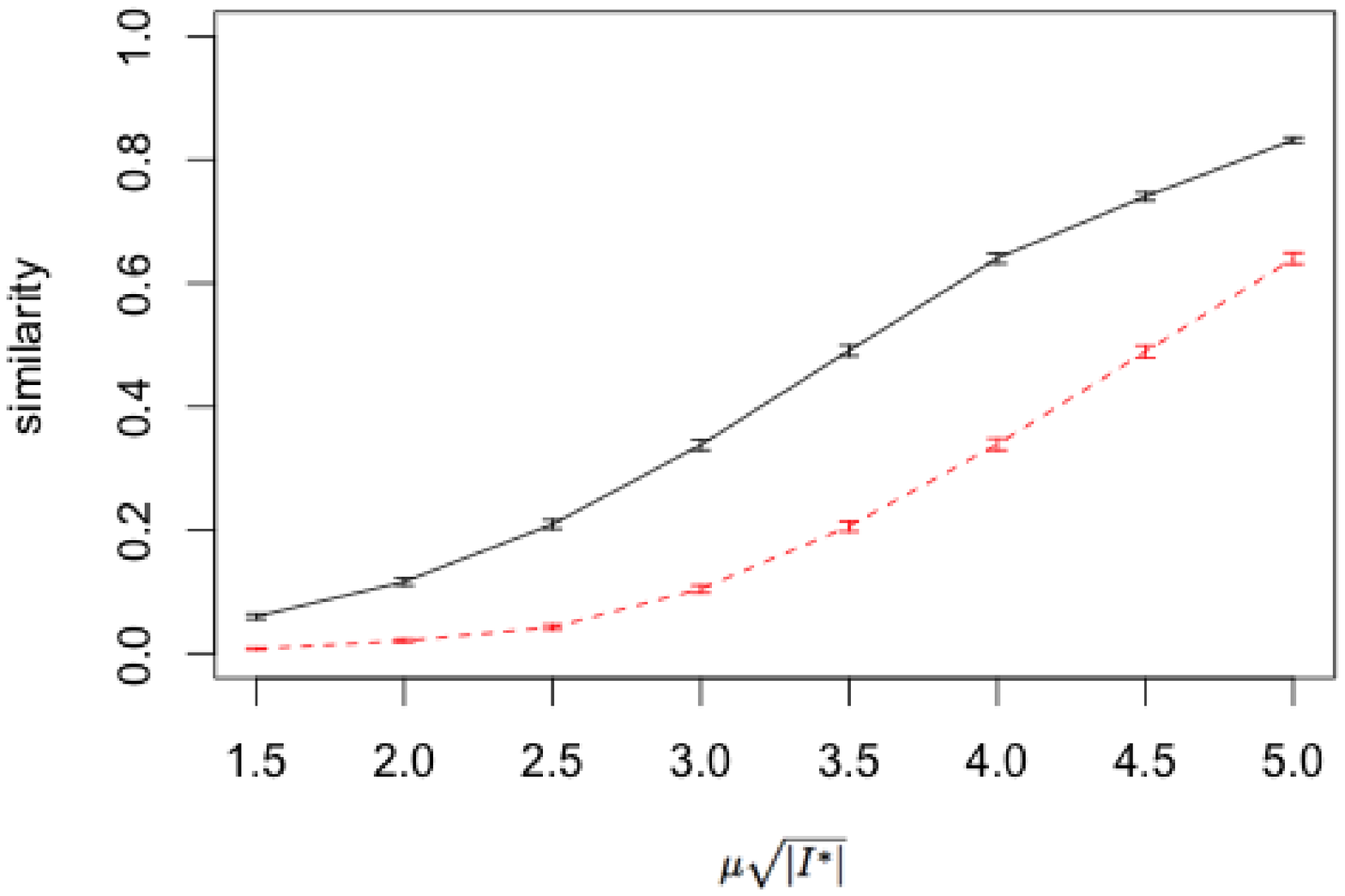}\includegraphics[scale=0.30]{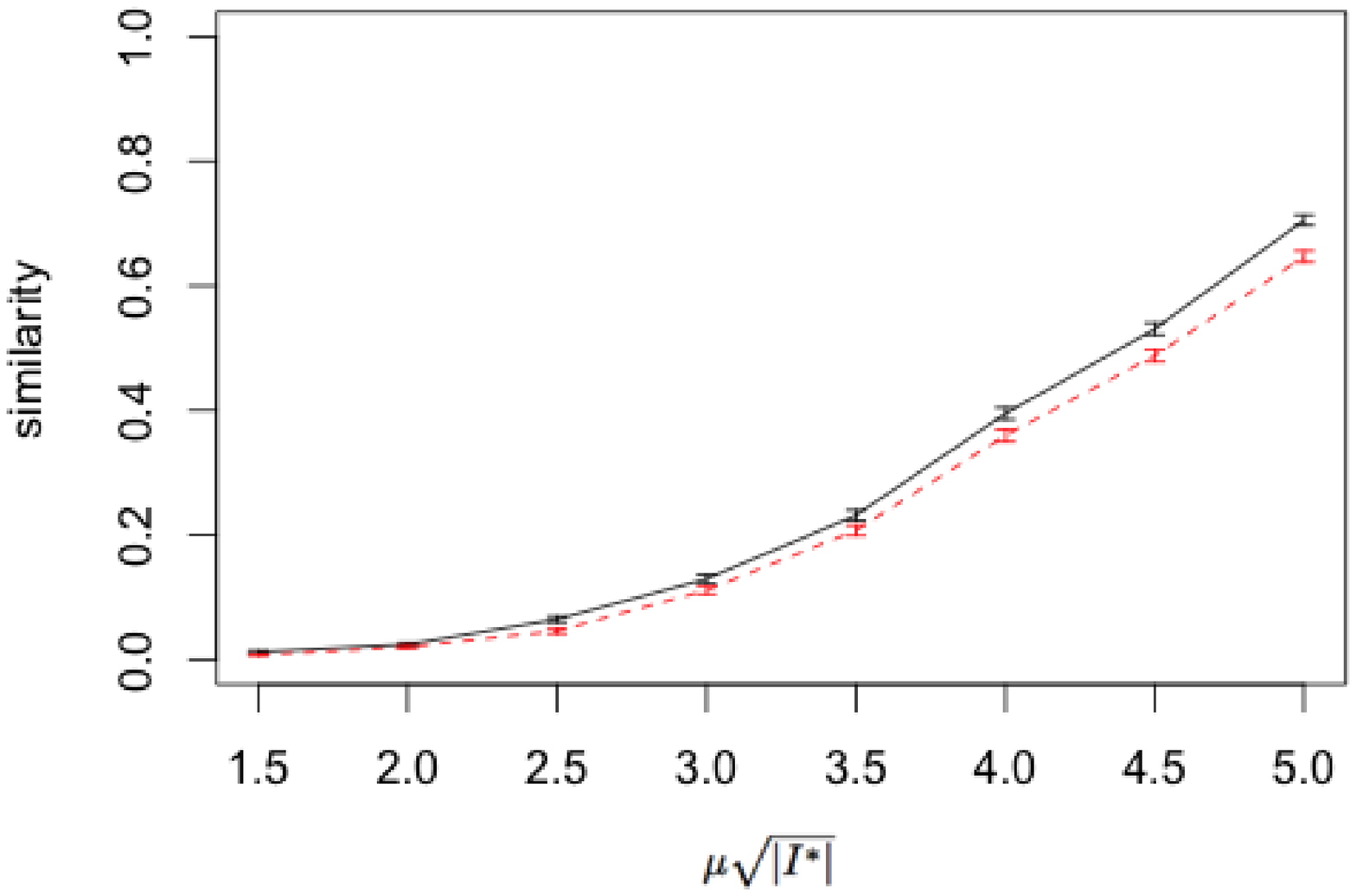}
\par\end{centering}
\end{figure}

For the second simulation, we give the similarities for different
choices of the ratio $n/|I_{n}^{*}|$. For each choice of $n/|I_{n}^{*}|$,
we choose $\mu  \sqrt{|I^*|} =1.2*\sqrt{2\log\frac{en}{|I_{n}^{*}|}}+0.1$. The simulation
result is shown in Figure \ref{fig:Similarity for different ratio}.
We can see that the $\mathcal{P}_{n}^{unpen}$ seems to be powerless
when $n/|I_{n}^{*}|$ is small ($|I_{n}^{*}|$ is large). However,
the gap between two procedures becomes smaller as the ratio $n/|I_{n}^{*}|$
increase.

\begin{figure}
\caption{\label{fig:Similarity for different ratio}Simulated similarities
for penalized procedure (black solid line) and unpenalized procedure
(red dashed line) under different ratios $n/|I_{n}^{*}|$ and $\mu \sqrt{|I^*|} =1.2*\sqrt{2\log\frac{en}{|I_{n}^{*}|}}+0.1$.
The $x$-axis is the ratio in log-scale and $y$-axis is the similarity.}

\begin{centering}
\includegraphics[scale=0.4]{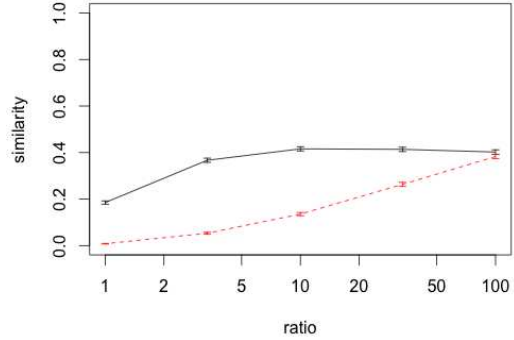}
\par\end{centering}

\end{figure}

We also compared the performance for Model  \eqref{eq:model-multi-dimension} in the two-dimensional case between our procedure $\mathcal{P}_n^{\{2\}}$ defined in Section \ref{sec:Multi-dimensional-block-signal} and $\mathcal{P}_n^{\{2\},unpen}$, which is defined exactly the same on the same approximation set except the penalty term is set to 0.  The simulation
result is shown in Figure \ref{fig:Similarity for two dimension}.  In this simulation, we give the
similarity for different choices of $\mu$ range from 2.5 to 6 with a step of 0.5. We choose $n=100$, so our space is a $100 \times 100$ rectangle. The left sub-graph gives the result for rectangular signal with width 30 and height 40 while the right sub-graph gives the result for rectangular signal with width 15 and height 80. Note that the area of the rectangle $|I^*_n| = 1200$ in both cases but the aspect ratios are different. We can see that $\mathcal{P}_n^{\{2\}}$ performs much better than $\mathcal{P}_n^{\{2\},unpen}$ in both cases and the performance is robust with respect to different aspect ratios.

\begin{figure}
\caption{\label{fig:Similarity for two dimension} Simulated similarities for different
$\mu$ in a $100 \times 100$ rectangle. The left sub-graph gives the result for $|I_{n}^{*}|=30 \times 40$
and the right sub-graph gives the result for $|I_{n}^{*}|=15 \times 60$. In
both graphs, the black solid line denotes the penalized procedure
$\mathcal{P}_{n}^{\{2\}}$ and the right dashed line denotes the unpenalized
procedure $\mathcal{P}_n^{\{2\},unpen}$. The $x$-axis is $\mu \sqrt{|I^*|}$ and $y$-axis is similarity.}
\begin{centering}
\includegraphics[scale=0.33]{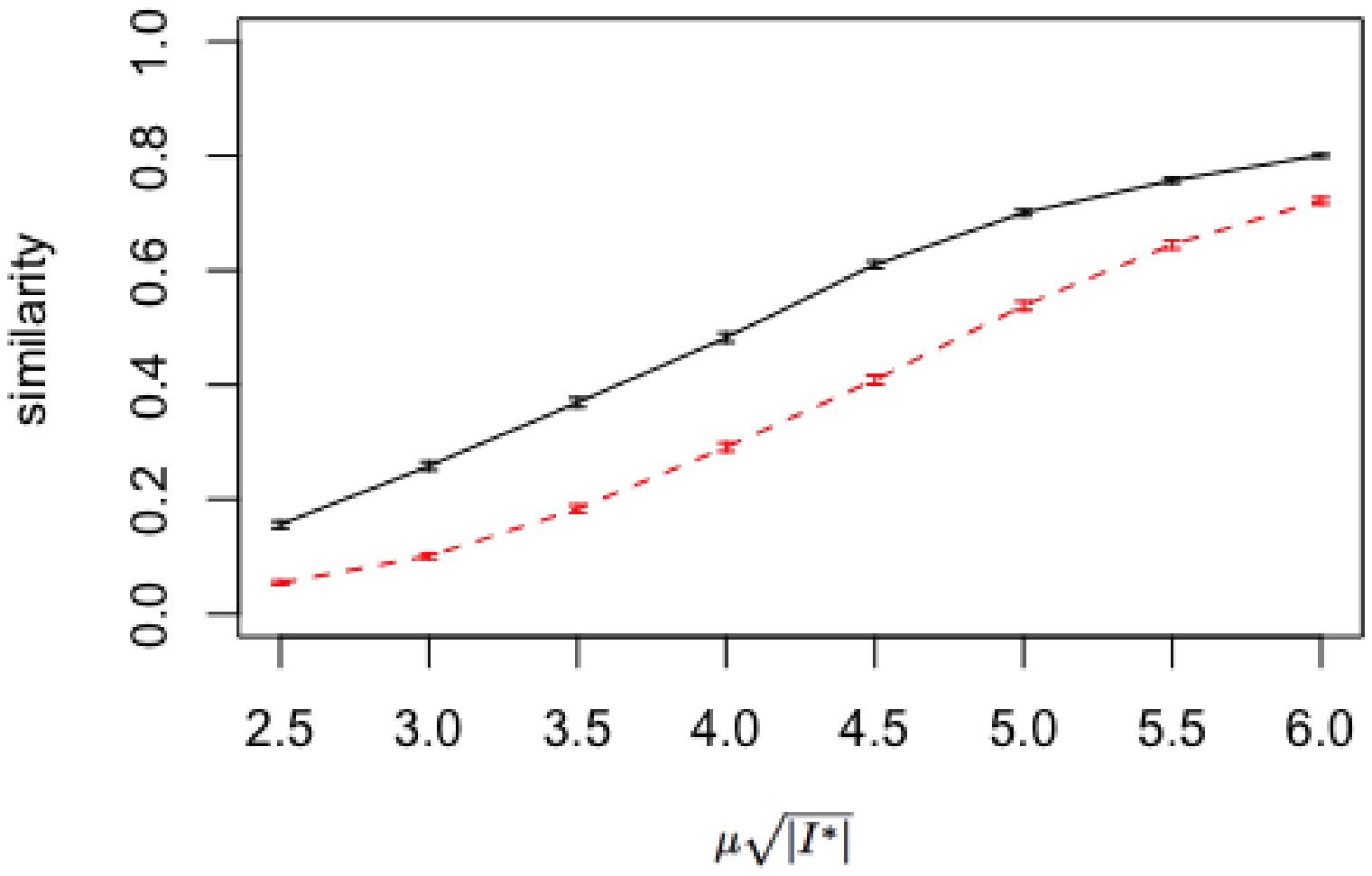}\includegraphics[scale=0.33]{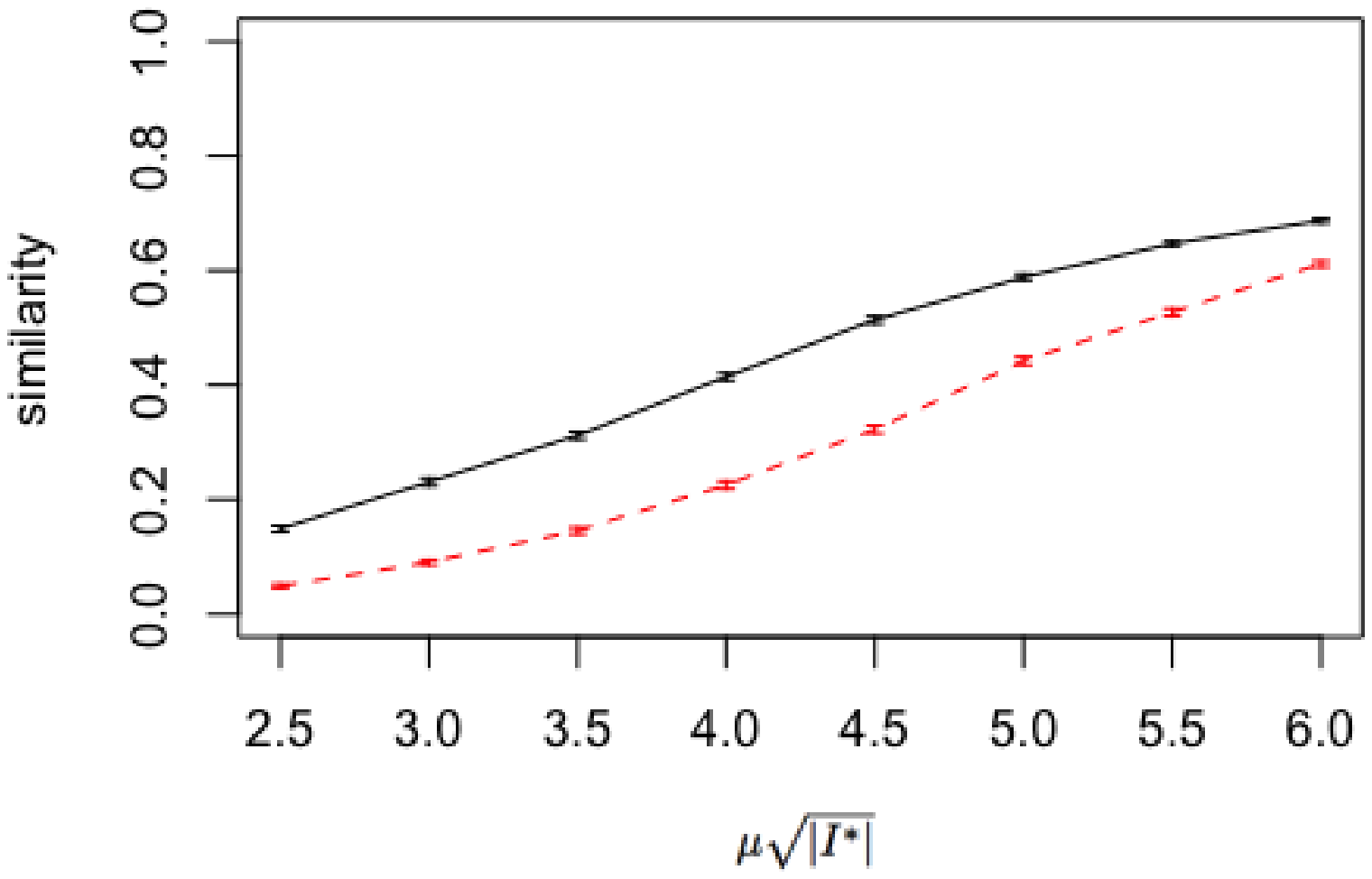}
\par\end{centering}
\end{figure}

All similarities in Figure \ref{fig: Similarity vs mu}, \ref{fig:Similarity for different ratio} and \ref{fig:Similarity for two dimension} are with respect to a 5\% significance level. 
The critical values were simulated with
10000 Monte Carlo samples, and the similarities were simulated with 2000
Monte Carlo samples. The location of the signal was sampled at random
in each of these simulations to avoid confounding with the approximation scheme.

\section{Extensions and Discussion\label{sec:Discussion}}

\subsection{Identification under the exponential family}

Now we consider the block signal identification problem under the exponential family
setting. If instead of Gaussian noise model, we observe independent
random variables $Y_{i}$, $i=1,\ldots,n$ through the one-dimensional
exponential family model $Y_{i}\sim F_{\nu_{i}}$, $i=1,\ldots,n$,
where 
\[
\nu_{i}=a+b\mathbf{1}_{I_{n}^{*}}(i)
\]
with baseline $a$ known, signal strength $b$ unknown, and unknown
interval $I_{n}^{*}$ defined the same as in Model \eqref{eq:model}.
The task is to recover the support of $I_{n}^{*}$. If $F$ is the
standard Gaussian distribution with $a=0$, we get back to our Gaussian
noise Model \eqref{eq:model}. Notice that if $|I|$ is sufficiently
large, then under $H_{0}$, $\boldsymbol{Y}(I)$ would be approximately
normally distributed, which suggest the optimality results in Section
\ref{sec:Block-identification-simple} would still hold provided $|I_{n}^{*}|$
is large enough. Formally, if $|I_{n}^{*}|\ge\log^{3+\delta}n$ for
some $\delta>0$ and denote $\mathcal{I}_{app}^{exp}=\mathcal{I}_{app}\text{\ensuremath{\cap}}\{I:|I|\ge\log^{3+\delta}n\}$,
then our identification procedure in Section \ref{sec:Block-identification-simple}
is consistent by replacing $\mathcal{I}_{app}$ with $\mathcal{I}_{app}^{exp}$.
For some similar arguments, see \cite{frick2014multiscale} and \cite{arias2011detection}. Alternatively, one can use $\sqrt{2 \log \boldsymbol{T}(I)}$ instead of $\boldsymbol{Y}(I)$ where $\boldsymbol{T}(I)$ is the local likelihood ratio statistic for testing $H_0: b=0$ against $H_1: b\ne 0$ on interval $I$.

\subsection{Identification with unspecified noise distribution}

In \cite{tony2012robust}, the author consider the identification
problem with an unspecified noise distribution. The idea is to apply
the identification procedure LRS on the ``local median transformed''
data. For more details about local median transformation, we refer
the reader to \cite{brown2008robust}. It is worthwhile to note that
our procedures $\mathcal{P}_{n}$, $\mathcal{P}_{n}^{\{2\}}$, $\mathcal{P}_{n,multi}$
are also adapted to the local median transformation and would give a
near-optimal solution over a broad range of noise distribution with
a much milder assumption on the length of the signal.

\subsection{Disucssion}

It is also interesting to compare our results with other results in change-point detection settings. For most research in change-point detection area, they typically seek to find an rate optimal solution rather than an exact optimal solution, due to the more complex structure they consider, see for example \cite{frick2014multiscale} and \cite{brunel2014convex}. Our signal identification problem has a slightly easier setting and we can achieve the exact optimal constant. We have shown in the simulation that the constant actually matters and a suboptimal constant may lead to a significant loss of power.
Last but not the least, in Theorem \ref{thm:(Lower-bound)-Multi}, we assume that the number
of block signals $K=O(\log^{p}n)$ for some $p>0$. If instead we assume
$K=n^{\delta}$ for some $\delta>0$, then our procedure may not be optimal anymore. In fact, in this case, block signal identification would be statistically more difficult
than the block signal detection, as there are so many block signals. It would be interesting to develop
an optimality theory under this situation, which is left as an open
problem.

\section*{Acknowledgements}
The author thanks Professor G. Walther for supervision and some helpful suggestions. We also owe thanks to the referees for a number of corrections and suggestions.

\section*{Appendix}

\subsection*{A.1 Proofs for Section \ref{sec:Block-identification-simple}}

Before giving the proof of Theorem \ref{thm: Main theorem}, we first
introduce some useful lemmas.

The following lemma is proved in \cite{chan2011detection}, which
is a consequence of a result in \cite{dumbgen_multiscale_2001}.
\begin{lem}
\label{lem: Dumbgen}Define $\boldsymbol{Z}_{n}(I)=\frac{\sum_{i\in I}Z_{i}}{\sqrt{|I|}}$.
Let $J\in\mathcal{I}_{n}$, where $\mathcal{I}_{n}=\{(j,k],0\le j<k\le n\}$
and $J$ does not depend on $\boldsymbol{Z}_{n}$. Then 
\[
\max_{I\in\mathcal{I}_{n}:I\subset J}\left(\boldsymbol{Z}_{n}(I)-\sqrt{2\log\frac{e|J|}{|I|}}\right)\stackrel{d}{=}:\, L<\infty\quad a.s.
\]
here the random variable L defined above is universally applicable
for all $J$ and n and is finite almost surely.
\end{lem}

Throughout the proof of this paper, we ignore the rounding issues
in the definition of $d_{\ell}$ whenever this does not affect our
results. To simplify the notation, in the following proof we use $I^{*}$ rather
than $I_{n}^{*}$ to denote the true signal whenever this does not
cause confusion.

The following lemma shows that we can approximate every interval well
using our approximation set $\mathcal{I}_{app}$ defined in
Section \ref{sec:Block-identification-simple}.
\begin{lem}
\label{lem:Approximation_One_dim}For each $I^{*}$ with $|I^{*}|\ll n$,
there exists an $\tilde{I}\in\mathcal{I}_{app}$ such that $D(\tilde{I},I^{*})\le\frac{1}{3\sqrt{\log_{2}\frac{n}{|I^{*}|}}}$
and $\sqrt{2\log\frac{en}{|\tilde{I}|}}-\sqrt{2\log\frac{en}{|I^{*}|}}=o(1)$.\end{lem}
\begin{proof}
There are two cases: $|I^{*}|>m_{\ell_{max}}$ and $|I^{*}|\le m_{\ell_{max}}$.

Consider first when $|I^{*}| > m_{\ell_{max}}$. Let $\ell^{*}$
be the integer satisfying $m_{\ell^{*}}<|I^{*}|\le2m_{\ell^{*}}$.
By our construction, unless $|I^{*}|<m_{\ell^{*}}+2d_{\ell^{*}}$,
there exists an interval $\tilde{I}\in\mathcal{I}_{app}$, such that
$\tilde{I}\text{\ensuremath{\subset}}I^{*}$ and 
\[
|\tilde{I}|>(1-\frac{2d_{\ell^{*}}}{m_{{\ell}^{*}}})|I^{*}|=(1-\frac{1}{3\sqrt{\ell^{*}}})|I^{*}|\text{\ensuremath{\ge}}(1-\frac{1}{3\sqrt{\log_{2}\frac{n}{|I^{*}|}}})|I^{*}|.
\]
Thus 
\[
D(\tilde{I},I^{*})=1-\sqrt{\frac{|\tilde{I}|}{|I^{*}|}}\le1-\sqrt{1-\frac{1}{3\sqrt{\log_{2}\frac{n}{|I^{*}|}}}}\le\frac{1}{3\sqrt{\log_{2}\frac{n}{|I^{*}|}}},
\]
 and 
\begin{eqnarray*}
0 & \le & \sqrt{2\log\frac{en}{|\tilde{I}|}}-\sqrt{2\log\frac{en}{|I^{*}|}}\\
 & \le & \sqrt{2\log\frac{en}{|I^{*}|(1-\frac{1}{3\sqrt{\log_{2}\frac{n}{|I^{*}|}}})}}-\sqrt{2\log\frac{en}{|I^{*}|}}\\
 & \le & \frac{-\log(1-\frac{1}{3\sqrt{\log_{2}\frac{n}{|I^{*}|}}})}{\sqrt{2\log\frac{en}{|I^{*}|}}}
\end{eqnarray*}
As we assume $|I^{*}|\ll n$, then numerator goes to zero and the
denominator goes to infinity, so $\sqrt{2\log\frac{en}{|\tilde{I}|}}-\sqrt{2\log\frac{en}{|I^{*}|}}=o(1)$.\\
If $ $$|I^{*}|<m_{\ell^{*}}+2d_{\ell^{*}}$, then there exists an
interval $\tilde{I}\in\mathcal{I}_{app}$, such that $\tilde{I}\text{\ensuremath{\supset}}I^{*}$, 
\[
|\tilde{I}|<(1+\frac{2d_{\ell^{*}}}{m_{\ell}^{*}})|I^{*}|=(1+\frac{1}{3\sqrt{\ell^{*}}})|I^{*}|\text{\ensuremath{\le}}(1+\frac{1}{3\sqrt{\log_{2}\frac{n}{|I^{*}|}}})|I^{*}|.
\]
The remaining proof for bounding $D(\tilde{I},I^{*})$ and $\sqrt{2\log\frac{en}{|\tilde{I}|}}-\sqrt{2\log\frac{en}{|I^{*}|}}$
are similar to the case $\tilde{I}\subset I^{*}$.

Now consider when $|I^{*}|\le m_{\ell_{max}}$ , then $I^{*}\in\mathcal{I}_{app}$,
and thus we can simply let $\tilde{I}=I^{*}$ and the theorem holds
trivially.
\end{proof}

\subsubsection*{Proof of Theorem \ref{thm: Main theorem}:}
\begin{proof}
Denote $\boldsymbol{X}(I)=\boldsymbol{Y}(I)-\sqrt{2\log\frac{en}{|I|}}$
for interval $I$. 

By the definition of $\gamma_{n}(\alpha)$,
\[
P_{H_{0}}(I\ne\emptyset)=P_{H_{0}}(\max_{I\in\mathcal{I}_{app}}\boldsymbol{X}(I)>\gamma_{n}(\alpha))\le\alpha,
\]
which proves equation \eqref{eq:type I error}.

Now turn to prove equation \eqref{eq: type II error}. Define
\[
K_{0}=\{I\in\mathcal{I}_{app}:I\cap I^{*}=\emptyset\}
\]
and 
\[
K_{1}=\{I\in\mathcal{I}_{app}:I\cap I^{*}\neq\emptyset\text{ and }D(I,I^{*})>\delta_{n}\}.
\]
By Lemma \ref{lem:Approximation_One_dim}, there exists an $\tilde{I}\in\mathcal{I}_{app}$
such that $D(\tilde{I},I^{*})\le\frac{1}{3\sqrt{\log_{2}\frac{n}{|I^{*}|}}}$
and $\sqrt{2\log\frac{en}{|\tilde{I}|}}-\sqrt{2\log\frac{en}{|I^{*}|}}=o(1)$.
Then,
\begin{eqnarray*}
P_{H_{1}}(D(\hat{I},I^{*})>\delta_{n}) & \le & P_{H_{1}}(\max(\max_{I\in\mathcal{I}_{app,}D(I,I^{*})>\delta_{n}}\boldsymbol{X}(I),\gamma_{n}(\alpha))\ge\boldsymbol{X}(\tilde{I}))\\
 & \le & P_{H_{1}}(\max_{I\in K_{0}}\boldsymbol{X}(I)\ge\boldsymbol{X}(\tilde{I}))+P_{H_{1}}(\max_{I\in K_{1}}\boldsymbol{X}(I)\ge\boldsymbol{X}(\tilde{I}))+P_{H_{1}}(\gamma_{n}(\alpha)\ge\boldsymbol{X}(\tilde{I})).
\end{eqnarray*}
By Lemma \ref{lem: Dumbgen}, $\max_{I\in K_{0}}\boldsymbol{X}(I)\stackrel{d}{\le}L<\infty$
a.s. and $\limsup_{n \rightarrow \infty} \gamma_{n}(\alpha)<\infty$ a.s.. Also notice that under
$H_{1}$, by Lemma \ref{lem:Approximation_One_dim}, we have 
\begin{eqnarray*}
\boldsymbol{X}(\tilde{I}) & = & \boldsymbol{Y}(\tilde{I})-\sqrt{2\log\frac{en}{|\tilde{I}|}}\\
 & \ge & -|\boldsymbol{Z}(\tilde{I})|+(1-D(\tilde{I},I^{*}))(\sqrt{2\log\frac{en}{|I^{*}|}}+b_{n})-\sqrt{2\log\frac{en}{|\tilde{I}|}}\\
 & \ge & -|\boldsymbol{Z}(\tilde{I})|+(1-D(\tilde{I},I^{*}))b_{n}-(\sqrt{2\log\frac{en}{|\tilde{I}|}}-\sqrt{2\log\frac{en}{|I^{*}|}})-D(\tilde{I},I^{*})\sqrt{2\log\frac{en}{|I^{*}|}}\\
 & \ge & -|\boldsymbol{Z}(\tilde{I})|+(1-D(\tilde{I},I^{*}))b_{n}-\frac{\sqrt{2\log\frac{en}{|I^{*}|}}}{3\sqrt{\log_{2}\frac{n}{|I^{*}|}}}-o(1)\\
 & \stackrel{p}{\rightarrow} & \infty,
\end{eqnarray*}
so 
\begin{equation}
P_{H_{1}}(\max_{I\in K_{0}}\boldsymbol{X}(I)\ge\boldsymbol{X}(\tilde{I}))\rightarrow0 \label{eq:first_part_small}
\end{equation}
and 
\begin{equation}
P_{H_{1}}(\gamma_{n}(\alpha)\ge\boldsymbol{X}(\tilde{I}))\rightarrow0. \label{eq:second_part_small}
\end{equation}
Denote 
\[
K_{near}=\{I\in K_{1}:\frac{|I|}{\log n}\le|I^{*}|\le|I|\log n\}.
\]
To finish the proof, it remains to show that $P_{H_{1}}(\max_{I\in K_{1}}\boldsymbol{X}(I)\ge\boldsymbol{X}(\tilde{I}))\rightarrow0$.
First notice that
\begin{eqnarray*}
P_{H_{1}}(\max_{I\in K_{1}}X(I)\ge X(\tilde{I})) & \le & \sum_{I\in K_{near}}P_{H_{1}}(\boldsymbol{X}(I)\ge\boldsymbol{X}(\tilde{I}))+P_{H_{1}}(\max_{I\in K_{1}\backslash K_{near}}\boldsymbol{X}(I)\ge\boldsymbol{X}(\tilde{I}))\\
 & =: & (A)+(B)
\end{eqnarray*}

We need the following two lemmas to bound part (A) and (B), respectively. The proof of these lemmas is given after this theorem.

\begin{lem}
\label{lem: lemma1_for_main_thm}
\[
\sum_{I\in K_{near}}P_{H_{1}}(\boldsymbol{X}(I)\ge\boldsymbol{X}(\tilde{I})) \rightarrow 0.
\]
\end{lem}

\begin{lem}
\label{lem: lemma2_for_main_thm}
\[
P_{H_{1}}(\max_{I\in K_{1}\backslash K_{near}}\boldsymbol{X}(I)\ge\boldsymbol{X}(\tilde{I})) \rightarrow 0.
\]
\end{lem}

Combining Lemma \ref{lem: lemma1_for_main_thm} and \ref{lem: lemma2_for_main_thm}, we see that $P_{H_{1}}(\max_{I\in K_{1}}X(I)\ge X(\tilde{I})) \rightarrow 0$. This together with Equations \eqref{eq:first_part_small} and \eqref{eq:second_part_small} will lead to $P_{H_{1}}(D(\hat{I},I^{*})>\delta_{n}) \rightarrow 0$ and we finish our proof.

\end{proof}

\subsubsection*{Proof of Lemma \ref{lem: lemma1_for_main_thm}:}
\begin{proof}
For each interval $I\in K_{near}$, we have
\begin{eqnarray*}
 &  & P_{H_{1}}(\boldsymbol{X}(I)\ge\boldsymbol{X}(\tilde{I}))\\
 & = & P_{H_{1}}(\boldsymbol{Y}(I)-\boldsymbol{Y}(\tilde{I})\ge\sqrt{2\log\frac{en}{|I|}}-\sqrt{2\log\frac{en}{|\tilde{I}|}}).
\end{eqnarray*}
Simple calculation shows that $\boldsymbol{Y}(I)-\boldsymbol{Y}(\tilde{I})$
has a normal distribution with mean
\[
-\mu \sqrt{|I^{*}|}(D(I,I^{*})-D(\tilde{I},I^{*})),
\]
and variance
\[
2D(I,\tilde{I})\le2.
\]
Thus, 
\begin{eqnarray*}
 &  & \sum_{I\in K_{near}}P_{H_{1}}(\boldsymbol{X}(I)\ge\boldsymbol{X}(\tilde{I}))\\
 & \le & \sum_{I\in K_{near}}\text{\ensuremath{\bar{\Phi}}}\left(\left(\left( \mu \sqrt{|I^{*}|}(D(I,I^{*})-D(\tilde{I},I^{*})) \right)-\sqrt{2\log\frac{en}{|\tilde{I}|}}+\sqrt{2\log\frac{en}{|I|}}\right)/\sqrt{2}\right)
\end{eqnarray*}
where $\bar{\Phi}$ denotes the upper cumulative distribution function
of the normal distribution.\\
Under the assumption that $\delta_{n}\gg \sqrt{\log\log n} / \sqrt{\log{n}}$ and $b_n \rightarrow +\infty$, we have
\[
(\delta_{n}-\frac{1}{3\sqrt{\log_{2}\frac{n}{|I^{*}|}}})(\sqrt{2\log\frac{en}{|I^{*}|}}+b_{n})\gg\sqrt{\log\log n}.
\]
Note that $D(I,I^{*})>\delta_{n}$ in $K_{1}$ and $D(\tilde{I},I^{*})\le\frac{1}{3\sqrt{\log_{2}\frac{n}{|I^{*}|}}}$ and $\mu \sqrt{|I^{*}|} \ge \sqrt{2\log\frac{en}{|I^{*}|}}+b_{n}$
so
\begin{equation}
\mu \sqrt{|I^{*}|} (D(I,I^{*})-D(\tilde{I},I^{*})) \gg\sqrt{\log\log n}.\label{eq: phi is small -1}
\end{equation}
When $I\in K_{near}$, we have $\frac{|I|}{\log n}\le|I^{*}|\le|I|\log n$,
then 
\begin{eqnarray*}
\left|\sqrt{2\log\frac{en}{|I|}}-\sqrt{2\log\frac{en}{|\text{\ensuremath{\tilde{I}}}|}}\right| & \le & \left|\sqrt{2\log\frac{en}{|I^{*}|}}-\sqrt{2\log\frac{en}{|\text{\ensuremath{\tilde{I}}}|}}\right|+\left|\sqrt{2\log\frac{en}{|I^{*}|}}-\sqrt{2\log\frac{en}{|I|}}\right|\\
 & \le & \left|\sqrt{2\log\frac{en}{|I^{*}|}}-\sqrt{2\log\frac{en}{|\text{\ensuremath{\tilde{I}}}|}}\right|+\frac{2\log\log n}{\sqrt{2\log\frac{en}{|I^{*}|}}}.
\end{eqnarray*}
Since we assume there exists a $\kappa>0$, such that $|I^{*}|\ll n^{1-\kappa}$,
then $\frac{2\log\log n}{\sqrt{2\log\frac{en}{|I^{*}|}}}\rightarrow0$.
By Lemma \ref{lem:Approximation_One_dim}, $\left|\sqrt{2\log\frac{en}{|I^{*}|}}-\sqrt{2\log\frac{en}{|\text{\ensuremath{\tilde{I}}}|}}\right|=o(1)$.
\begin{equation}
\left|\sqrt{2\log\frac{en}{|I|}}-\sqrt{2\log\frac{en}{|\text{\ensuremath{\tilde{I}}}|}}\right|=o(1).\label{eq: Phi is small -2}
\end{equation}

Combine \eqref{eq: phi is small -1}, \eqref{eq: Phi is small -2}
and using the inequality $\bar{\Phi}(x)\le\exp(-\frac{x^{2}}{2})$
for $x>1$, we have,
\[
\text{\ensuremath{\bar{\Phi}}}\left(\left(\left( \mu \sqrt{|I^{*}|}(D(I,I^{*})-D(\tilde{I},I^{*}))\right)-\sqrt{2\log\frac{en}{|\tilde{I}|}}+\sqrt{2\log\frac{en}{|I|}}\right)/\sqrt{2}\right)\le\log^{-\eta}n
\]
for all $\eta>0$. \\
Simple counting shows the cardinality of the set $K_{near}$ is $O(\log^{2}n)$.
Thus 
\[
\sum_{I\in K_{near}}P_{H_{1}}(\boldsymbol{X}(I)\ge\boldsymbol{X}(I^{*}))\rightarrow0.
\]

\end{proof}

\subsubsection*{Proof of Lemma \ref{lem: lemma2_for_main_thm}:}
\begin{proof}
When $I\in K_{1}\backslash K_{near}$, we have
$1-D(I,I^{*})\le\min(\sqrt{\frac{|I|}{|I^{*}|}},\sqrt{\frac{|I^{*}|}{|I|}})\le1/\sqrt{\log n}$.

Thus,
\begin{eqnarray*}
 &  & P_{H_{1}}(\max_{I\in K_{1}\backslash K_{near}}\boldsymbol{X}(I)\ge\boldsymbol{X}(\tilde{I}))\\
 & = & P_{H_{1}}(\max_{I\in K_{1}\backslash K_{near}}(\boldsymbol{Z}(I)+(1-D(I,I^{*}))(\sqrt{2\log\frac{en}{|I^{*}|}}+b_{n})-\sqrt{2\log\frac{en}{|I|}})\\
 &  & \qquad\ge\boldsymbol{Z}(\tilde{I})+(1-D(\tilde{I},I^{*}))(\sqrt{2\log\frac{en}{|I^{*}|}}+b_{n})-\sqrt{2\log\frac{en}{|\tilde{I}|}}))\\
 & \le & P_{H_{1}}(\max_{I\in K_{1}\backslash K_{near}}(\boldsymbol{Z}(I)-\sqrt{2\log\frac{en}{|I|}}+\frac{1}{\sqrt{\log n}}(\sqrt{2\log\frac{en}{|I^{*}|}}+b_{n}))\\
 &  & \qquad\ge\boldsymbol{Z}(\tilde{I})-\sqrt{2\log\frac{en}{|\tilde{I}|}}+(1-\frac{1}{3\sqrt{\log_{2}\frac{n}{|I^{*}|}}})(\sqrt{2\log\frac{en}{|I^{*}|}}+b_{n}))\\
 & \le & P_{H_{1}}(\max_{I\in K_{1}\backslash K_{near}}(\boldsymbol{Z}(I)-\sqrt{2\log\frac{en}{|I|}})\\
 &  & \qquad\ge\boldsymbol{Z}(\tilde{I})-(\sqrt{2\log\frac{en}{|\tilde{I}|}}-\sqrt{2\log\frac{en}{|I^{*}|}})\\
 &  & \qquad\quad-(\frac{1}{3\sqrt{\log\frac{n}{|I^{*}|}}}+\frac{1}{\sqrt{\log n}})(\sqrt{2\log\frac{en}{|I^{*}|}})+(1-\frac{1}{3\sqrt{\log_{2}\frac{n}{|I^{*}|}}}-\frac{1}{\sqrt{\log n}})b_{n})
\end{eqnarray*}
Since $ $$\max_{I\in K_{1}\backslash K_{near}}(\boldsymbol{Z}(I)-\sqrt{2\log\frac{en}{|I|}})=O_{p}(1)$
by Lemma \ref{lem: Dumbgen}, $\boldsymbol{Z}(\tilde{I})=O_{p}(1)$,
$(\frac{1}{3\sqrt{\log_{2}\frac{n}{|I^{*}|}}}+\frac{1}{\sqrt{\log n}})(\sqrt{2\log\frac{en}{|I^{*}|}})=O(1)$,
$\sqrt{2\log\frac{en}{|\tilde{I}|}}-\sqrt{2\log\frac{en}{|I^{*}|}}=o(1)$
by Lemma \ref{lem:Approximation_One_dim} and $(1-\frac{1}{3\sqrt{\log_{2}\frac{n}{|I^{*}|}}}-\frac{1}{\sqrt{\log n}})b_{n}\rightarrow\infty$,
then $P_{H_{1}}(\max_{I\in K_{1}\backslash K_{near}}\boldsymbol{X}(I)\ge\boldsymbol{X}(\tilde{I}))\rightarrow0$
and we finish our proof.

\end{proof}

\subsection*{A.2 Proofs for Section \ref{sec:Multi-dimensional-block-signal}}

\subsubsection*{Proof of Proposition \textmd{\ref{prop: High-dim-null-distribution}}:}
\begin{proof}
Denote \#$\mathcal{I}_{app}^{\{2\}}(\ell)$ be the number of rectangles
in $\mathcal{I}_{app}^{\{2\}}(\ell)$. A simple counting shows that
when $\ell\le\lfloor\log_{2}\frac{n^{2}}{\log n}\rfloor$, $\#I_{app}^{\{2\}}(\ell)=2\ell^{3}2^{\ell}$.
Let's abuse the notation a bit: for $\lceil\log_{2}n^{2}\rceil\ge\ell>\lfloor\log_{2}\frac{n^{2}}{\log n}\rfloor$,
let $I_{app}^{\{2\}}(\ell)=\{I\in I_{small},n^{2}2^{-\ell}<|I|\le2n^{2}2^{-\ell}\}$.
Then simple counting again gives $\#I_{app}^{\{2\}}(\ell)\le n^{2}\sum_{|I|=n^{2}2^{-\ell}}^{2n^{2}2^{-\ell}}|I|\le2n^{4}2^{-\ell}\log n\le2\log^{3}n2^{\ell}\le2\ell^{3}2^{\ell}$,
where the the second to the last inequality comes from $\ell>\lfloor\log_{2}\frac{n^{2}}{\log n}\rfloor$,
so $2^{2l}\ge\frac{n^{4}}{\log^{2}n}$ and last inequality comes from
$\ell>\lfloor\log_{2}\frac{n^{2}}{\log n}\rfloor\ge\log n$. Thus
for all $1\le\ell\le\lceil\log_{2}n^{2}\rceil$, we have $\#I_{app}^{\{2\}}(\ell)\le2\ell^{3}2^{\ell}$.

For $\kappa>2$, we obtain:
\begin{eqnarray*}
 &  & P(\max_{I\in I_{app}^{\{2\}}}(\boldsymbol{Y}(I)-\sqrt{2\log\frac{en^{2}}{|I|}})>\kappa)\\
 & \le & \sum_{\ell=1}^{\lceil\log_{2}n^{2}\rceil}\#I_{app}^{\{2\}}(\ell)\max_{I\in I_{app}(\ell)}\exp(-\frac{1}{2}(\sqrt{2\log\frac{en^{2}}{|I|}}+\kappa)^{2})\\
 & \le & \sum_{\ell=1}^{\lceil\log_{2}n^{2}\rceil}2\ell^{3}2^{\ell}2^{-\ell}\exp(-\kappa\sqrt{\ell}-\kappa^{2}/2)\\
 & = & \sum_{\ell=1}^{\lceil\log_{2}n^{2}\rceil}2\ell^{3}\exp(-\kappa\sqrt{\ell}-\kappa^{2}/2)\\
 & \le & C\exp(-\kappa^{2}/2)
\end{eqnarray*}
for some constant $C>0$ not depending on $n$. Thus, we have 
\[
\max_{I\in I_{app}^{\{2\}}}(\boldsymbol{Y}(I)-\sqrt{2\log\frac{en^{2}}{|I|}}) = O_p(1).
\]
\end{proof}

Analogously to Lemma \ref{lem:Approximation_One_dim}, the following
lemma shows that in the two-dimensional case, we can also approximate
every rectangle well enough by $\mathcal{I}_{app}^{\{2\}}$.
\begin{lem}
\label{lem:Approximation_Mul_dim}For each $I^{*}$ with $|I^{*}|\ll n^{2}$,
there exists an $\tilde{I}\in\mathcal{I}_{app}^{\{2\}}$ such that
$D(\tilde{I},I^{*})\le\frac{6}{\sqrt{\log_{2}\frac{n^{2}}{|I^{*}|}}}$
and $\sqrt{2\log\frac{en^{2}}{|\tilde{I}|}}-\sqrt{2\log\frac{en^{2}}{|I^{*}|}}=o(1)$.
\end{lem}
The proof of this lemma is very similar to Lemma \ref{lem:Approximation_One_dim},
and thus is omitted. See also \cite{walther2010optimal}.

\subsubsection*{Proof of Theorem \ref{thm: Main theorem-high-dim}: }
\begin{proof}
Denote $\boldsymbol{X}(I)=\boldsymbol{Y}(I)-\sqrt{2\log\frac{en^{2}}{|I|}}$
for rectangle $I$. 

By the definition of $\gamma_{n}^{\{2\}}(\alpha)$,
\[
P_{H_{0}}(\hat{I}\ne\phi)=P_{H_{0}}(\max_{I\in\mathcal{I}_{app}}\boldsymbol{X}(I)\ge\gamma_{n}^{\{2\}}(\alpha))\le\alpha,
\]
which proves \eqref{eq:type I error}.

Now we turn to prove \eqref{eq: type II error}. Again, let $I^{*}$
denote the true rectangle with length $a^{*}$ and width $b^{*}$. Define
\[
K_{0}=\{I\in\mathcal{I}_{app}^{\{2\}}:I\cap I^{*}=\emptyset\}
\]
and 
\[
K_{1}=\{I\in\mathcal{I}_{app}^{\{2\}}:I\cap I^{*}\neq\emptyset\text{ and }D(I,I^{*})>\delta_{n}\}.
\]
By Lemma \ref{lem:Approximation_Mul_dim}, there exists an rectangle
$\tilde{I}\in\mathcal{I}_{app}^{\{2\}}$, such that $D(\tilde{I},I^{*})\le\frac{6}{\sqrt{\log_{2}\frac{n^{2}}{|I^{*}|}}}$
and $\sqrt{2\log\frac{en^{2}}{|\tilde{I}|}}-\sqrt{2\log\frac{en^{2}}{|I^{*}|}}=o(1)$.
$ $Then,
\begin{eqnarray*}
P_{H_{1}}(D(\hat{I},I^{*})>\delta_{n}) & \le & P_{H_{1}}(\max(\max_{I\in\mathcal{I}_{app,}^{\{2\}}D(I,I^{*})>\delta_{n}}\boldsymbol{X}(I),\gamma_{n}^{\{2\}}(\alpha))\ge\boldsymbol{X}(\tilde{I}))\\
 & \le & P_{H_{1}}(\max_{I\in K_{0}}\boldsymbol{X}(I)\ge\boldsymbol{X}(\tilde{I}))+P_{H_{1}}(\max_{I\in K_{1}}(\boldsymbol{X}(I)\ge\boldsymbol{X}(\tilde{I}))+P_{H_{1}}(\gamma_{n}^{\{2\}}(\alpha)\ge\boldsymbol{X}(\tilde{I})).
\end{eqnarray*}
By Proposition \ref{prop: High-dim-null-distribution}, $\max_{I\in K_{0}}\boldsymbol{X}(I)=O_{p}(1)$
and $\gamma_{n}^{\{2\}}(\alpha)<\infty$ a.s.. Notice that under $H_{1}$,
by Lemma \ref{lem:Approximation_Mul_dim} and the same reasoning as
in the proof of Theorem \ref{thm: Main theorem}, we have $X(\tilde{I})\stackrel{p}{\rightarrow}\infty$.
Thus $P_{H_{1}}(\max_{I\in K_{0}}\boldsymbol{X}(I)\ge X(\tilde{I}))\rightarrow0$
and $P_{H_{1}}(\gamma_{n}^{\{2\}}(\alpha)\ge\boldsymbol{X}(\tilde{I}))\rightarrow0$. 

Now consider the term $P_{H_{1}}(\max_{I\in K_{1}}(\boldsymbol{X}(I)\ge\boldsymbol{X}(\tilde{I}))$.
Denote 
\[
K_{near}=\{I\in K_{1}:\frac{a^{*}}{\log n}\le a\le a^{*}\log n\text{ and }\frac{b^{*}}{\log n}\le b\le b^{*}\log n\},
\]
then
\begin{eqnarray*}
P_{H_{1}}(\max_{I\in K_{1}}X(I)\ge X(\tilde{I})) & \le & \sum_{I\in K_{near}}P_{H_{1}}(\boldsymbol{X}(I)\ge\boldsymbol{X}(\tilde{I}))+P_{H_{1}}(\max_{I\in K_{1}\backslash K_{near}}\boldsymbol{X}(I)\ge\boldsymbol{X}(\tilde{I}))\\
 & =: & (A)+(B)
\end{eqnarray*}

Consider part (A) first. Similar to the proof of Theorem \ref{thm: Main theorem},
one can show that for all $I\in K_{near}$
\[
P_{H_{1}}(\boldsymbol{X}(I)\ge\boldsymbol{X}(\tilde{I}))\le\log^{-\eta}n
\]
for all $\eta>0$. Simple counting shows the cardinality of the set
$K_{near}$ is $O(\log^{4}n)$. Thus 
\[
\sum_{I\in K_{near}}P_{H_{1}}(\boldsymbol{X}(I)\ge\boldsymbol{X}(I^{*}))\rightarrow0.
\]

Consider part (B), in this case, $1-D(I,I^{*})\le1/\sqrt{\log n}$.
Thus
\begin{eqnarray*}
 &  & P_{H_{1}}(\max_{I\in K_{1}\backslash K_{near}}\boldsymbol{X}(I)\ge\boldsymbol{X}(\tilde{I}))\\
 & = & P_{H_{1}}(\max_{I\in K_{1}\backslash K_{near}}(\boldsymbol{Z}(I)+(1-D(I,I^{*}))(\sqrt{2\log\frac{en^{2}}{|I^{*}|}}+b_{n})-\sqrt{2\log\frac{en^{2}}{|I|}})\\
 &  & \quad\ge\boldsymbol{Z}(\tilde{I})+(1-D(\tilde{I},I^{*}))(\sqrt{2\log\frac{en^{2}}{|I^{*}|}}+b_{n})-\sqrt{2\log\frac{en^{2}}{|\tilde{I}|}}))\\
 & \le & P_{H_{1}}(\max_{I\in K_{1}\backslash K_{near}}(\boldsymbol{Z}(I)-\sqrt{2\log\frac{en^{2}}{|I|}}+\frac{1}{\sqrt{\log n}}(\sqrt{2\log\frac{en^{2}}{|I^{*}|}}+b_{n}))\\
 &  & \quad\ge\boldsymbol{Z}(\tilde{I})-\sqrt{2\log\frac{en^{2}}{|\tilde{I}|}}+(1-\frac{6}{\sqrt{\log_{2}\frac{n^{2}}{|I^{*}|}}})(\sqrt{2\log\frac{en^{2}}{|I^{*}|}}+b_{n}))\\
 & \le & P_{H_{1}}(\max_{I\in K_{1}\backslash K_{near}}(\boldsymbol{Z}(I)-\sqrt{2\log\frac{en^{2}}{|I|}})\ge\boldsymbol{Z}(\tilde{I})-(\sqrt{2\log\frac{en^{2}}{|\tilde{I}|}}-\sqrt{2\log\frac{en^{2}}{|I^{*}|}})\\
 &  & \quad-(\frac{6}{\sqrt{\log\frac{n^{2}}{|I^{*}|}}}+\frac{1}{\sqrt{\log n}})(\sqrt{2\log\frac{en^{2}}{|I^{*}|}})+(1-\frac{6}{\sqrt{\log_{2}\frac{n^{2}}{|I^{*}|}}}-\frac{1}{\sqrt{\log n}})b_{n})
\end{eqnarray*}

Since $ $$\max_{I\in K_{1}\backslash K_{near}}(\boldsymbol{Z}(I)-\sqrt{2\log\frac{en^{2}}{|I|}})=O_{p}(1)$
by Proposition \ref{prop: High-dim-null-distribution}, $\boldsymbol{Z}(\tilde{I})=O_{p}(1)$,
$(\frac{6}{\sqrt{\log_{2}\frac{n^{2}}{|I^{*}|}}}+\frac{1}{\sqrt{\log n}})(\sqrt{2\log\frac{en^{2}}{|I^{*}|}})=O(1)$,
$(\sqrt{2\log\frac{en^{2}}{|\tilde{I}|}}-\sqrt{2\log\frac{en^{2}}{|I^{*}|}})=o(1)$
by Lemma \ref{lem:Approximation_Mul_dim}, and $(1-\frac{6}{\sqrt{\log_{2}\frac{n^{2}}{|I^{*}|}}}-\frac{1}{\sqrt{\log n}})b_{n}\rightarrow\infty$,
then $P_{H_{1}}(\max_{I\in K_{1}\backslash K_{near}}\boldsymbol{X}(I)\ge\boldsymbol{X}(\tilde{I}))\rightarrow0$.
\end{proof}

\subsection*{A.3 Proof for Section \ref{sec:Signal-identification-More-Than-One}}

We need the following lemma in the proof of Theorem \ref{thm: More than one signal}.
\begin{lem}
\label{lem: Tail_Of_L}Denote $L=\max_{I\in I_{app}}(\boldsymbol{Y}(I)-\sqrt{2\log\frac{en}{|I|}})$.
Let $Z$ be a standard normal random variable, not necessarily independent
with $L$, then for all $\kappa>4$, there exists a constant $C>0$
not depending on $n$ and $\kappa$ such that 
\[
P(L+Z>\kappa)\le C\exp(-\kappa^{2}/8)
\]
\end{lem}
\begin{proof}
Similar to the proof of Proposition \ref{prop: High-dim-null-distribution},
we know that when $\kappa>2$, $P(L>\kappa)\le C^{'}\exp(-\kappa^{2}/2)$
for some $C^{'}>0$ not depending on $n$ and $\kappa$. Thus,

\begin{eqnarray*}
P(L+Z>\kappa) & \le & P(L>\kappa/2)+P(Z>\kappa/2)\\
 & \le & C^{'}\exp(-\kappa^{2}/8)+2\exp(-\kappa^{2}/8)\\
 & \le & C\exp(-\kappa^{2}/8).
\end{eqnarray*}

\end{proof}

\subsubsection*{Proof of Theorem \ref{thm: More than one signal}:}
\begin{proof}
As before, denote $\boldsymbol{X}(I)=\boldsymbol{Y}(I)-\sqrt{2\log\frac{en}{|I|}}$
for interval $I$. 

When there exists no signal, by the definition of $M_{n}$,
\[
P_{H_{0}}(I\ne\emptyset)=P_{H_{0}}(\max_{I\in\mathcal{I}_{app}}\boldsymbol{X}(I)>\gamma_{n}^ {})\le \alpha,
\]
which proves \eqref{eq: Type I of multiple signals}.

Now we turn to prove \eqref{eq: Consistency of multiple signals}.
It is enough to show that with probability approaching 1, we will
not stop before the $Kth$ iteration and for each of the first $K$
iterations, we can correctly identify one of the true signals with precision $\delta_n$.

Recall the true signals $I^{*}=\{I_{1}^{*},I_{2}^{*},\ldots,I_{K}^{*}\}$.
By Lemma \ref{lem:Approximation_One_dim}, for each $j=1,\ldots,K$,
there exists an interval $\tilde{I_{j}}$, such that $D(\tilde{I_{j}},I_{j}^{*})\le\frac{1}{3\sqrt{\log_{2}\frac{n}{|I_{j}^{*}|}}}$
and $\sqrt{2\log\frac{en}{|\tilde{I}_{j}|}}-\sqrt{2\log\frac{en}{|I_{j}^{*}|}}=o(1)$.

Consider the event 
\[
E_{1}=\{\max(\max_{I\in\mathcal{I}_{app,}^{1}\min_{j=1}^{K}D(I,I_{j}^{*})>\delta_{n}}\boldsymbol{X}(I),\gamma_{n}(\alpha))<\min_{j=1}^{K}\boldsymbol{X}(\tilde{I_{j}})\}:=\{LHS<RHS\}.
\]
If $E_{1}$ holds, then we can be sure that the interval $\hat{I}_{1}$
identified by the first iteration satisfies $D(\hat{I}_{1},I_{j_{1}}^{*})<\delta_{n}$
for some $j_{1}$ in 1 to $K$ and $\hat{I}_{1}\cap I_{j}^{*}=\emptyset$
for all $j\ne j_{1}$ by the assumption of $d_{min}$. After the first
iteration, consider the event 
\[
E_{2}=\{\max(\max_{I\in\mathcal{I}_{app,}^{2}\min_{j=1,j\ne j_{1}}^{K}D(I,I_{j}^{*})>\delta_{n}}\boldsymbol{X}(I),\gamma_{n}(\alpha))<\min_{j=1,j\ne j_{1}}^{K}\boldsymbol{X}(\tilde{I_{j}})\}:=\{LHS<RHS\},
\]
the LHS of $E_{2}$ is non-increasing while the RHS of $E_{2}$ is
non-decreasing compared to those of $E_{1}$. Thus, if $E_{1}$ holds,
$E_{2}$ must hold, and we can be sure that the interval $\hat{I}_{2}$
identified by the second iteration satisfies $D(\hat{I}_{2},I_{j_{2}}^{*})\le\delta_{n}$
for some $j_{2}$ from 1 to $K$, $j_{2}\ne j_{1}$ and $\hat{I}_{2}\cap I_{j}^{*}=\emptyset$
for all $j\ne j_{2}$. If this procedure can be repeated for $K$ times,
then we can identify all $K$ signals with precision $\delta_{n}$.
Thus, 
\begin{eqnarray*}
 &  & P_{H_{1}}(\max_{j=1}^{K}\min_{i = 1}^{\hat{K}} D(\hat{I_{i}},I_{j}^{*})>\delta_{n})\\
 & \le & P_{H_{1}}(E_{1}^{C})\\
 & = & P_{H_{1}}(\max(\max_{I\in\mathcal{I}_{app,}^{1}\min_{j=1}^{K}D(I,I_{j}^{*})>\delta_{n}}\boldsymbol{X}(I),\gamma_{n}(\alpha))\ge\min_{j=1}^{K}\boldsymbol{X}(\tilde{I_{j}}))
\end{eqnarray*}
Define
\[
K_{0}=\{I\in\mathcal{I}_{app}:I\cap I_{j}^{*}=\emptyset\text{ for all }j=1,\ldots,K\}
\]
and 
\[
K_{1}=\{I\in\mathcal{I}_{app}:I\cap I_{j}^{*}\neq\emptyset\text{ for some }j=1,\ldots,K\text{ and }\min_{j=1}^{K}D(I,I_{j}^{*})>\delta_{n}\}.
\]
Then 
\[
\begin{aligned} & P_{H_{1}}(\max(\max_{I\in\mathcal{I}_{app,}^{1}\min_{j=1}^{K}D(I,I_{j}^{*})>\delta_{n}}\boldsymbol{X}(I),\gamma_{n}(\alpha))\ge\min_{j=1}^{K}\boldsymbol{X}(\tilde{I_{j}}))\\
\le & P_{H_{1}}(\max_{I\in K_{0}}\boldsymbol{X}(I)\ge\min_{j=1}^{K}\boldsymbol{X}(\tilde{I}{}_{j}))+P_{H_{1}}(\max_{I\in K_{1}}\boldsymbol{X}(I)\ge\min_{j=1}^{K}\boldsymbol{X}(\tilde{I_{j}}))+P_{H_{1}}(\gamma_{n}(\alpha)\ge\min_{j=1}^{K}\boldsymbol{X}(\tilde{I}{}_{j}))
\end{aligned}
\]

By Lemma \ref{lem: Dumbgen}, $\max_{I\in K_{0}}\boldsymbol{X}(I)\stackrel{d}{\le}L<\infty$
a.s. and $\gamma_{n}<\infty$ a.s.. Notice that under $H_{1}$ and our
assumption $b_{n,j}\gg\sqrt{\log\log n}$, each $X(\tilde{I}_{j})$
is Gaussian distributed with mean greater than
\[
(\sqrt{2\log\frac{en}{|I_{j}^{*}|}}+b_{n})(1-\frac{1}{3\sqrt{\log\frac{n}{|I_{j}^{*}|}}})-\sqrt{2\log\frac{en}{|\tilde{I}_{j}|}}\gg\sqrt{\log\log n}
\]
and variance 1. Then similar to the proof of Theorem \ref{thm: Main theorem},
for each $j=1,\ldots,K$, and $\eta>0$, we have 
\[
P_{H_{1}}(\max_{I\in K_{0}}\boldsymbol{X}(I)\ge\boldsymbol{X}(\tilde{I}{}_{j}))\le\log^{-\eta}n
\]
 and 
\[
P_{H_{1}}(\gamma_{n}(\alpha)\ge\boldsymbol{X}(\tilde{I}{}_{j}))\le\log^{-\eta}n.
\]
As a result, under the assumption that $K=\log^{p}n$ for some $p>0$,
$P_{H_{1}}(\max_{I\in K_{0}}\boldsymbol{X}(I)\ge\min_{j=1}^{K}\boldsymbol{X}(\tilde{I}{}_{j}))\le KP_{H_{1}}(\max_{I\in K_{0}}\boldsymbol{X}(I)\ge\boldsymbol{X}(\tilde{I}{}_{j}))\rightarrow0$
and $P_{H_{1}}(\gamma_{n}(\alpha)\ge\min_{j=1}^{K}\boldsymbol{X}(\tilde{I}{}_{j}))\le KP_{H_{1}}(\gamma_{n}(\alpha)\ge\boldsymbol{X}(\tilde{I}{}_{j}))\rightarrow0$.

We only need to show that $P_{H_{1}}(\max_{I\in K_{1}}\boldsymbol{X}(I)\ge\min_{j=1}^{K}\boldsymbol{X}(\tilde{I_{j}}))\rightarrow0$.
Denote 
\[
K_{near}=\{I\in K_{1}:\min_{j=1}^{K}D(I,I_{j}^{*})<1-1/\sqrt{\log n}\},
\]
then
\begin{eqnarray*}
P_{H_{1}}(\max_{I\in K_{1}}\boldsymbol{X}(I)\ge\min_{j=1}^{K}\boldsymbol{X}(\tilde{I_{j}})) & \le & \sum_{j=1}^{K}\sum_{I\in K_{near}}P_{H_{1}}(\boldsymbol{X}(I)\ge\boldsymbol{X}(\tilde{I}_{j}))+\sum_{j=1}^{K}P_{H_{1}}(\max_{I\in K_{1}\backslash K_{near}}\boldsymbol{X}(I)\ge\boldsymbol{X}(\tilde{I}))\\
 & =: & (A)+(B)
\end{eqnarray*}

For part (A), as in the proof of Theorem \ref{thm: Main theorem},
for each $j=1,\ldots,K$, $P_{H_{1}}(\boldsymbol{X}(I)\ge\boldsymbol{X}(\tilde{I}_{j}))\le\log^{-\eta}n$
for all $\eta>0$. Note that if any signal $I$ intersects with more
than one element in $I^{*}$, then we must have $\min_{j=1}^{K}D(I,I_{j}^{*})\ge1-1/\sqrt{\log n}$
by our assumption of $d_{min}$, and thus such $I\notin K_{near}$.
Thus, the cardinality of the set $K_{near}$ is $O(\log^{2}nK)=O(\log^{p+2}n)$.
As a result, 
\[
\sum_{j=1}^{K}\sum_{I\in K_{near}}P_{H_{1}}(\boldsymbol{X}(I)\ge\boldsymbol{X}(I_{j}^{*}))\rightarrow0.
\]

For part (B), in this situation $1-D(I,I_{j}^{*})\le1/\sqrt{\log n}$
for all $j=1,\ldots,K$. Thus
\begin{eqnarray*}
 &  & P_{H_{1}}(\max_{I\in K_{1}\backslash K_{near}}\boldsymbol{X}(I)\ge\boldsymbol{X}(\tilde{I}_{j}))\\
 & = & P_{H_{1}}(\max_{I\in K_{1}\backslash K_{near}}(\boldsymbol{Z}(I)+(1-D(I,I_{j}^{*}))(\sqrt{2\log\frac{en}{|I_{j}^{*}|}}+b_{n,j})-\sqrt{2\log\frac{en}{|I|}})\\
 &  & \qquad\ge\boldsymbol{Z}(\tilde{I}_{j})+(1-D(\tilde{I}_{j},I_{j}^{*}))(\sqrt{2\log\frac{en}{|I_{j}^{*}|}}+b_{n,j})-\sqrt{2\log\frac{en}{|\tilde{I}_{j}|}}))\\
 & \le & P_{H_{1}}(\max_{I\in K_{1}\backslash K_{near}}(\boldsymbol{Z}(I)-\sqrt{2\log\frac{en}{|I|}}+\frac{1}{\sqrt{\log n}}(\sqrt{2\log\frac{en}{|I_{j}^{*}|}}+b_{n,j}))\\
 &  & \qquad\ge\boldsymbol{Z}(\tilde{I}_{j})-\sqrt{2\log\frac{en}{|\tilde{I}_{j}|}}+(1-\frac{1}{3\sqrt{\log\frac{n}{|I_{j}^{*}|}}})(\sqrt{2\log\frac{en}{|I_{j}^{*}|}}+b_{n,j}))\\
 & \le & P_{H_{1}}(\max_{I\in K_{1}\backslash K_{near}}(\boldsymbol{Z}(I)-\sqrt{2\log\frac{en}{|I|}})\\
 &  & \qquad\ge\boldsymbol{Z}(\tilde{I}_{j})-(\sqrt{2\log\frac{en}{|\tilde{I}_{j}|}}-\sqrt{2\log\frac{en}{|I_{j}^{*}|}})\\
 &  & \qquad\quad-(\frac{1}{3\sqrt{\log\frac{n}{|I_{j}^{*}|}}}+\frac{1}{\sqrt{\log n}})(\sqrt{2\log\frac{en}{|I_{j}^{*}|}})+(1-\frac{1}{3\sqrt{\log\frac{n}{|I_{j}^{*}|}}}-\frac{1}{\sqrt{\log n}})b_{n,j}
\end{eqnarray*}

By Lemma \ref{lem: Tail_Of_L}, $P_{H_{1}}(\max_{I\in K_{1}\backslash K_{near}}(\boldsymbol{Z}(I)-\sqrt{2\log\frac{en}{|I|}})-Z(\tilde{I}_{j})>\kappa)\le C\exp(-\kappa^{2}/8)$.
Notice that $(\frac{1}{3\sqrt{\log\frac{n}{|I_{j}^{*}|}}}+\frac{1}{\sqrt{\log n}})(\sqrt{2\log\frac{en}{|I_{j}^{*}|}})=O(1)$
and $(1-\frac{1}{3\sqrt{\log\frac{n}{|I_{j}^{*}|}}}-\frac{1}{\sqrt{\log n}})b_{n,j}\gg\sqrt{\log\log n}$,
thus $P_{H_{1}}(\max_{\hat{I}\in K_{1}\backslash K_{near}}\boldsymbol{X}(I)\ge\boldsymbol{X}(\tilde{I}_{j}))\le\log^{-\eta}n$
for all $\eta>0$. So $\sum_{j=1}^{K}P_{H_{1}}(\max_{I\in K_{1}\backslash K_{near}}\boldsymbol{X}(I)\ge\boldsymbol{X}(\tilde{I}_{j}))\rightarrow0$.
As a result, $P_{H_{1}}(\max_{\hat{I}\in K_{1}}(\boldsymbol{X}(I)\ge\min_{j=1}^{K}\boldsymbol{X}(\tilde{I_{j}}))\rightarrow0$
and we finish our proof for \eqref{eq: Consistency of multiple signals}. 

Note that for each iteration, we may only remove an interval very close to $I_{j}^{*}$
for some $j=1,\ldots,K$. So after $K$ iteration, there may still
exist intervals $I\in\mathcal{I}_{app}^{K}$ such that $I\cap I_{j}^{*}\ne\emptyset$
for some $j$, denote these intervals by $\mathcal{I}_{left}$. By
similar argument as above, we can show that $P(\max_{I\in\mathcal{I}_{left}}X(I)>\gamma_{n}(\alpha))\rightarrow0$.
Now consider all intervals in $\mathcal{I}_{app}^{K}\backslash\mathcal{I}_{left}$,
By the definition of $\gamma_{n}(\alpha)$ and noticing the fact that
$\gamma_{n}(\alpha)$ is non-decreasing in $n$, we can conclude that
the number of the false intervals our procedure identifies is controlled
by a geometric distribution with parameter $\alpha$. Thus $E\hat{K}\le K+\frac{\alpha}{1-\alpha}+o(1)$
and \eqref{eq: Expect of multiple signals} follows by letting $C(\alpha)=\frac{\alpha}{1-\alpha}$.
\end{proof}

\subsubsection*{Proof of Theorem \ref{thm:(Lower-bound)-Multi}: }
\begin{proof}
Assuming without loss of generality that $\frac{n}{|I_{1}^{*}|}$
is an integer. Assume first that the signals can only start and end
in a grid given by $\{i|I_{1}^{*}|+1,\ldots(i+1)|I_{1}^{*}|\}$ for
$i=1,\ldots,\frac{n}{|I_{1}^{*}|}$. According to \cite{jeng2010optimal},
it is enough to show that Theorem \ref{thm:(Lower-bound)-Multi} holds
under this assumption. Let $R_{i}=\sum_{j=1}^{|I_{1}^{*}|}X_{i|I_{1}^{*}|+j}/\sqrt{|I_{1}^{*}|}=:r_{i}+Z_{i}^{'}$
for $i=0,\ldots,n^{'}-1$, where $n^{'}=\frac{n}{|I_{1}^{*}|}$. Then
$Z_{i}^{'}\stackrel{iid}{\sim}N(0,1)$, and $r_{i}=0$ for all but
$K$ locations, while at these locations, $r_{i}\le\sqrt{2\log\frac{en}{|I_{1}^{*}|}}=\sqrt{2\log en^{'}}$.
Since $\log K=b\log\log n=o(\log n)$, by Lemma 1 in \cite{jeng2010optimal},
no identification procedure can be consistent under this model. 
\end{proof}

\section*{References}

\bibliographystyle{abbrv}
\bibliography{reference}

\end{document}